\documentclass[reprint,3p,times,onecolumn,12pt]{elsarticle} 

\usepackage{fancyhdr}
\makeatletter
\def\ps@pprintTitle{%
	\let\@oddhead\@empty
	\let\@evenhead\@empty
	\def\@oddfoot{Journal Ref.: Chaos Solitons \& Fractals {\bf 155} (2022) 111652. \hfill DOI: \href{https://doi.org/10.1016/j.chaos.2021.111652}{10.1016/j.chaos.2021.111652}}%
	\let\@evenfoot\@oddfoot
}
\makeatother

\usepackage{amsmath,amsthm}
\usepackage{amsfonts}
\usepackage{amssymb}
\usepackage{hyperref}
\usepackage{graphicx}
\usepackage{latexsym}
\usepackage{amssymb}
\usepackage{amscd,amssymb}
\usepackage[arrow,matrix]{xy}
\usepackage{graphicx}
\usepackage{epstopdf}

\usepackage{soul}
\biboptions{sort&compress}
\usepackage{setspace}
\usepackage{amscd,amsmath,amsthm,amssymb,color}
\usepackage{xcolor}
\allowdisplaybreaks
\theoremstyle{plain}
\newtheorem{Theorem}{Theorem}[section]

\emergencystretch=\maxdimen
\newcommand{\bea}{\begin{eqnarray}}
	\newcommand{\eea}{\end{eqnarray}}
\newcommand{\bes}{\begin{subequations}}
	\newcommand{\ees}{\end{subequations}}

\begin{document}
\title{Localized Nonlinear Waves on Spatio-Temporally Controllable Backgrounds for a (3+1)-Dimensional Kadomtsev-Petviashvili-Boussinesq Model in Water Waves} 
	
\author[nit]{\large Sudhir Singh}
\author[apctp]{\large  K. Sakkaravarthi \corref{cor}}
\author[nit]{\large K. Murugesan}
	
\address[nit]{Department of Mathematics, National Institute of Technology, Tiruchirappalli -- 620015, Tamil Nadu, India}
\address[apctp]{Young Scientist Training Program, Asia-Pacific Center for Theoretical Physics (APCTP),\\ POSTECH Campus, Pohang -- 37673, Republic of Korea}

\cortext[cor]{Corresponding author Email address: ksakkaravarthi@gmail.com (K. Sakkaravarthi)} 


\setstretch{1.150}	
\begin{abstract} 
	The physics of nonlinear waves on variable backgrounds and the relevant mathematical analysis continues to be the challenging aspect of the study. In this work, we consider a (3+1)-dimensional nonlinear model describing the dynamics of water waves and construct nonlinear wave solutions on spatiotemporally controllable backgrounds for the first time by using a simple mathematical tool auto-B\"acklund transformation. Mainly, we unravel physically interesting features to control and manipulate the dynamics of nonlinear waves through the background. We construct a single kink-soliton and rogue wave, respectively, by adopting an exponential function and general polynomial of degree two as initial seed solutions. We choose arbitrary periodic, localized and combined wave backgrounds by incorporating Jacobi elliptic functions and investigate the modulation of these two nonlinear waves with a clear analysis and graphical demonstrations. The solutions derived in this work give us sufficient freedom to generate exotic nonlinear coherent structures on variable backgrounds and open up an interesting direction to explore the dynamics of various other nonlinear waves propagating through inhomogeneous media.
	
\end{abstract}

\maketitle
\setstretch{0.950}	
	{\bf Highlights}
	\begin{itemize}
		\item Dynamics of localized waves on controllable backgrounds in (3+1)D nonlinear model is investigated.
		\item Analytic nonlinear wave solutions with arbitrary backgrounds are constructed using a simplified approach.
		\item Periodic and localized wave backgrounds are incorporated by Jacobi elliptic functions.
		\item The impact of different sn-, cn-, and dn-type backgrounds in the kink-soliton and rogue waves are demonstrated. 
		\item Possible directions to extend the present analysis to several nonlinear models are presented.
	\end{itemize}

\setstretch{1.150}	
\section{Introduction}
Studies on understanding the dynamics of nonlinear waves have been fascinating for a couple of centuries. More intense exploration has been made in the past few decades with the help of computational tools to deal with theoretical models describing such waves and technologies to their experimental demonstrations. To be precise mathematical modelling, solving the model, analyzing the associated dynamics and performing experiments to simulate theoretical results are the continuous steps to investigate the features and identify possible applications of such nonlinear waves. These nonlinear waves appear in diverse fields from  hydrodynamics to advanced gravitational/atmospheric science \cite{Scott-book}. There are plenty of mathematical models, both in scalar (one component) and multicomponent vector systems in one and higher-dimensional versions, with several results on nonlinear wave propagation and different types of interactions. To highlight the fact, we wish to mention that sine-Gordon (sG), Korteweg-de Vries (KdV), Kadomtsev-Petviashvili (KP), Boussinesq, and nonlinear Schr\"odinger (NLS) family of equations and their generalizations are a few of them which arise in various fields as governing models and attracted huge attention among the researchers working on these fields \cite{Whitham-book, Yang-book}. {\textcolor{black}{ Importantly, the nature and evolution of solitary waves or solitons, breathers, rogue waves, lumps, dromions, positons, complexitons, peakons, etc. and interactions among themselves as well as with other waves are investigated in recent years by adapting well-established, modified, and several novel methods, starting from the much efficient inverse scattering transform \cite{Ablo-book} to Hirota bilinear method \cite{Hirota-book, r01, r08, r06}, Lie group analysis \cite{ r02}, Darboux transformation \cite{ r03, r04}, B\"acklund transformation \cite{Back-book}, optimal Galerkin-homotopy asymptotic method  \cite{r07}, and various other expansion methods  \cite{r05} and references therein}}. {\textcolor{black}{Among these techniques, the B\"acklund transformations is one of the very successful approach to obtain localized wave solutions due to its ease and wide applicability. This approach is being successfully applied to several nonlinear models including single-component (scalar) and vector/coupled (multi-component) systems in one-dimension ((1+1)D) as well as in higher-dimensional ((2+1)D and (3+1)D) equations with different types of nonlinearities \cite{r1,r2,r3,r4,bt-mmas21,bt-tian21,bt-tian17}, to mention a few. Especially, the nature of solitons, breathers, lump and rogue waves in higher-order (1+1)D Boussinesq-Burgers equation \cite{r1}, (2+1)D dispersive long-wave system \cite{r2}, (2+1)D coupled Burgers model \cite{r4}, (3+1)D variable-coefficient Kadomtsev-Petviashvili-Burgers-type equation \cite{r3}, (3+1)D shallow-water waves \cite{bt-mmas21}, (3+1)D generalized Kadomtsev-Petviashvili equation \cite{bt-tian21} and inhomogeneous coupled nonlinear Schr\"odinger system \cite{bt-tian17} are studied using B\"acklund transformation method.}} 

{\textcolor{black}{Most of the above discussed nonlinear waves appear independently on some constant/uniform backgrounds or without any background and portray a vast collection of wave phenomena. However, the possibility of such nonlinear wave occurrence on varying backgrounds with controllable features becomes an open question and it started to attract interest among the researchers due to the physical significance and realization in different contexts. Particularly, the physical motivation to look for such nonlinear waves on  non-uniform/varying backgrounds starts from the situation of randomly varying surface or deep water waves to inhomogeneous plasma, layered magnetic materials, inhomogeneous optical media, and atomic condensate system \cite{Boris-book,PRR20,pnas21,FiPexp}.}} {\textcolor{black}{ As a result of this search, some localized nonlinear waves on varying backgrounds are investigated in recent times, which include the rogue waves on cnoidal, periodic, and solitary wave backgrounds in one-dimensional models such as focusing NLS model \cite{epjst,PHD20,RSPA,SiAM}, derivative NLS equation \cite{PRE21,  NLS21,PS20}, higher-order nonlinear Schr\"odinger equation \cite{ Chaos21,PLA21}, higher-order modified KdV equation \cite{CNSNS21}, modified KdV models \cite{mmas21,PS21}, Hirota equation \cite{EPJP21,NLD20}, Gerdjikov-Ivanov model \cite{CSF19}, sine-Gordon equation \cite{AML20, RSPA20}, Fokas model \cite{AML21}, and coupled  cubic-quintic NLS equation \cite{FiP20} as well as vector Chen-Lee-Liu NLS model \cite{PRR21}. Mostly, the method used in these studies is nothing but the Darboux transformation which requires Lax pair and involves complex mathematical calculations.}} 

\textcolor{black}{Though, there exist significant works on nonlinear waves with different backgrounds in one-dimensional models as mentioned above, such an attempt for higher-dimensional systems becomes scarce and yet to be explored. Further, considering the experimental observations in (1+1)D models, it is also required to unravel the features of non-uniform/varying backgrounds in higher-dimensional systems too. So, we wish to pursue a possible mechanism to generate localized nonlinear waves on controllable backgrounds in higher-dimensional models theoretically.} 
For this purpose, we consider the following (3+1)-dimensional model referred to as Kadomtsev-Petviashvili-Boussinesq (KPB) equation, which consists of three spatial dimensions in addition to the time coordinate \cite{amwt, jpy}:
	\begin{equation}
		u_{xxxy}+3(u_{x}u_{y})_{x}+u_{ty}+u_{tx}-u_{zz}+u_{tt}=0. \label{e1}
	\end{equation}
\textcolor{black}{The above KPB model (\ref{e1}) describes the propagation and interaction of localized structures in shallow and deep water \cite{amwt}. It is a generalized version of Kadomtsev-Petviashvili (KP) and Boussinesq equations with spatial and the spatio-temporal dispersions in addition to second-order temporal evolution. Also, it is well known that the KP and Boussinesq models respectively explain the dynamics of shallow and deep water, and hence the present KP model coupled with Boussinesq effect describes the dynamics around an interface of both deep and shallow water waves in (3+1)D which are represented as right and left moving waves \cite{jpy}. Further, this KPB model (\ref{e1}) can provide an insight into physical modeling due to additional Boussinesq term as the zero-mass assumption is not required \cite{sun}.} 
In recent years, the KPB model (\ref{e1}) attracted several researchers to study its different perspectives including the integrability nature, symmetry solutions, conservation laws, and certain nonlinear wave solutions using bilinear/B\"acklund transformation. {Though the integrability in higher dimensional models is a challenging task, the present KPB model (\ref{e1}) is an integrable generalization with mixed KP and Boussinesq effects. Note that the dispersion relations of the KP and KPB models are quite different. However, the considered model ensures the Hirota bilinear form and so various physically interesting nonlinear wave structures}. 
{\textcolor{black}{Looking at the literature on the KPB model  (\ref{e1}), we can find that the solutions of one- and two-solitons using the simplified Hirota method \cite{amwt}, traveling waves using bilinear B\"acklund transformations \cite{jpy}, high-order breathers and rogue waves \cite{sun} and higher-order rogue waves with generalized polynomials \cite{ wli} and lump and interaction waves \cite{mplb} through Hirota's bilinear method are reported. Also, the Painlev\'e integrability analysis \cite{lkaur}, localized wave solutions using bilinear form \cite{lkaur2} and the symmetry reductions along with conserved quantities are obtained \cite{ ldm}.}
\textcolor{black}{Recently, Manafian has reported multi-rogue wave solutions using generalized polynomials and reduced bilinear form along with kink-soliton solutions using multiple exp-function method for the present KPB model \cite{mma21}.} 
For further detailed analysis and discussion, one can refer to these reports and references therein.

\textcolor{black}{Note that the solutions reported in the above-mentioned works \cite{amwt, sun, lkaur, mplb, jpy, lkaur2, wli, ldm,mma21}, especially the soliton and rogue wave solutions reported in \cite{amwt,sun,mma21}, are on constant (zero/non-zero) background, while the solutions we intend to discuss here is on variable backgrounds. Considering the limited works for localized waves on variable backgrounds in higher-dimensional models, especially no such report is available for the present KPB model (\ref{e1}), we proceed to pursue our interest to construct localized nonlinear wave solutions on spatio-temporally controllable backgrounds and study their dynamics in detail.} \textcolor{black}{For this purpose, we implement an auto-B\"acklund transformation 
as explained in Sec. \ref{sec2}. We construct  analytical solutions for kink-soliton and rational rogue wave in Sec. \ref{sec3} and explain the notion of controllable wave structures for different combinations of backgrounds. Section \ref{sec4} is devoted for providing a categorical discussions on some important results obtained in the work and the final section \ref{sec5} contains formal conclusions with significance of the results.} 

\section{Methodology to Construct Nonlinear Waves Solutions on Background}\label{sec2}
In this section, we deduce an auto-B\"acklund transformation which shall generate physically interesting solution structures explaining the nature of various localized nonlinear waves by using the truncated Painlev\'e approach. \textcolor{black}{Here we wish to mention that most of the previously known analyses for nonlinear waves on varying backgrounds (given in the Introduction) involve complex mathematical calculations like Darboux transformation and also require that the models should admit Lax pair. In pursuit of an easier and more adaptable tool to construct solutions with controllable backgrounds, we have succeeded when attempting with this truncated Painlev\'e approach and extracted auto-B\"acklund transformation. Thus we propose this `simple mathematical tool' to obtain  analytical solutions for constructing various types of nonlinear waves and for investigating their evolutionary characteristics under a wider possible backgrounds. Also, we believe that this procedure can be straightforwardly adapted for different other nonlinear models of potential interest.} We know that the KPB equation \eqref{e1} is Painlev\'e integrable and possesses the Laurent series truncation \cite{lkaur} defined as $u(x,y,z,t) = u_0 (x,y,z,t) + \dfrac{u_1(x,y,z,t)}{\phi(x,y,z,t)}$.
{\textcolor{black}{On substituting this into KPB equation \eqref{e1}, we get a polynomial equation with terms at different powers of $\phi(x,y,z,t)$ as $\sum_{j=1}^{6} \Omega_{j}\phi^{j-6}= 0,$ thereby 
	solving the equations arising at each power of $\phi$, we obtain the explicit auto-B\"acklund transformation. However, in general, exactly solving this set of equations is another tedious task because of the higher-order nonlinearity. To overcome this difficulty and to prelude the effectiveness of the methodology, we consider a special case of truncation in the Laurent series by replacing $u_0 (x,y,z,t)$ with  $u_0 (z,t)$ which has one temporal and spatial effect. This special truncation will enhance the ease of computation and provides an entry to arbitrary functions in the solution as described in Theorem $2.1$.}}
	
	\textcolor{black}{Now, considering the motivation to study the importance of spatio-temporal background $u_0$ and the above description, we look for a solution of (3+1)D KPB model (\ref{e1}) in the following form:}
	\begin{equation}
		u(x,y,z,t) = u_0 (z,t) + \dfrac{u_1(x,y,z,t)}{\phi(x,y,z,t)}, \label{eq14}
	\end{equation} 
	where $u_0$ is any function depending on one space $z$ and time $t$ coordinates, while it is independent from other two spatial dimensions $x$ and $y$ (restricted for mathematical simplicity as mentioned above). \textcolor{black}{Further, one can consider the general function $u_0 (x,y,z,t)$ with all spatial dimensions and shall proceed computing the required solutions, which we leave as a future assignment.} 
	
	\begin{Theorem}
		Let $\phi(x,y,z,t)$ is a solution of the two coupled linear and nonlinear equations\vspace{-0.280cm}
		\bes\bea 
		&&\phi_{xtt} - \phi_{xzz} + \phi_{xyt} + \phi_{xxt} + \phi_{xxxxy}  = 0, \label{eq105} \\
		&&\phi_t (\phi_x + \phi_y ) - 3 \phi_{xy} \phi_{xx} + 3 \phi_x \phi_{xxy} + \phi_y \phi_{xxx} + \phi_t ^2 - \phi_z^2 = 0,\label{eq106}\eea 
		{then} \vspace{-0.5cm}
		\bea
		&&u(x,y,z,t) = f(t-z) + g(t+z) + \dfrac{2 \phi _x }{\phi },\qquad\qquad \label{eq16}
		\eea \ees 
		is the solution of KPB equation \eqref{e1}, where the sum of the first two terms involving arbitrary functions $f(t-z)$ and $g(t+z)$ is nothing but the general solution of wave equation $u_{0,tt} - u_{0,zz}=0$.
	\end{Theorem}
	\begin{proof}
		On substituting \eqref{eq14} into KPB equation \eqref{e1}, we obtain the following expression:
		\bes \bea \Omega _{1} \phi ^{-5} + \Omega_{2} \phi ^{-4} + \Omega _{3}\phi ^{-3} + \Omega_{4}\phi ^{-2} + \Omega _{5} \phi ^{-1} + \Omega_{6} =0, \eea
		where \vspace{-0.50cm}
		\bea 
		\Omega_{1} &=& -12 u_1 ^2 \phi _x ^2 \phi _y + 24 u_1 \phi _x ^3 \phi _y ,\label{eq67}\\
		\Omega_{2} &=& 15 u_1 u_{1,x} \phi _x \phi _y  + 9 u_1 u_{1,y} \phi _x ^2 - 18 u_{1,x} \phi _x ^2 \phi _y - 6 u_{1,y}\phi _x ^ 3 + 3 u_1 ^2 \phi _x \phi _{xy}- 18 u_1 \phi _x ^2 \phi _{xy} \nonumber\\ &&  + 3 u_1 ^2  \phi _{xx} \phi _y- 18 u_1 \phi _x \phi _y \phi _{xx} ,   \label{eq68}\\
		\Omega _{3} &=& 2 u_1 \phi _t ^ 2 - 2 u_1 \phi _z^2 + 2 u_1 \phi _y \phi _t - 3 u_{1,x}^2 \phi _y + 2 u_1 \phi_x \phi _t - 9 u_{1,x} u_{1,y} \phi _x  - 3 u_1 u_{1,xy} \phi _x \nonumber\\ 
		&&+ 6 u_{1,xy} \phi _x^2 - 3 u_1 u_{1,x} \phi _{xy} + 12 u_{1,x} \phi _x \phi _{xy} - 3 u_1 u_{1,xx} \phi _y + 6 u_{1,xx} \phi _x \phi _y - 3 u_1 u_{1,y}\phi _{xx} \nonumber\\ 
		&&+ 6 u_{1,x}\phi _{xx} \phi _y + 6 u_{1,y}\phi _x \phi _{xx} + 6 u_1 \phi _{xy} \phi _{xx} + 6 u_1 \phi _x \phi _{xxy} + 2 u_1 \phi _{xxx} \phi _y, \label{eq69}\\
		\Omega _{4} &=& -2 u_{1,t} \phi _t - u_1 \phi _{tt} + 2 u_{1,z} \phi _z + u_1 \phi _{zz} - u_{1,y} \phi _t - u_{1,t} \phi_y - u_1 \phi _{yt} - u_{1,x} \phi _t - u_{1,t} \phi _x \nonumber\\ 
		&&- u_1 \phi _{xt} + 3 u_{1,x}u_{1,xy} + 3 u_{1,y} u_{1,xx} - 3 u_{1,xx} \phi _{xy} - 3 u_{1,xy} \phi _{xx} - 3 u_{1,xxy}\phi _x - 3 u_{1,x} \phi _{xxy} \nonumber\\ 
		&&- u_{1,xxx}\phi _y - u_{1,y} \phi _{xxx} - u_{1} \phi _{xxxy},  \label{eq70}\\
		\Omega_{5} &=& u_{1,tt}- u_{1,zz} + u_{1,yt} + u_{1,xt} + u_{1,xxxy},  \label{eq71}\\
		\Omega_{6} &=& u_{0,tt}- u_{0,zz}.  \label{eq73}
		\eea \ees 
		\textcolor{black}{To obtain the required solution, we need to solve the above system of equations recursively arising at each power of $\phi(x,y,z,t)$. On solving for $\Omega_{1} =0$ from Eq. (\ref{eq67}), we get $u_1= 2 \phi _x $ and this straightforwardly satisfies $\Omega_{2}=0$ from (\ref{eq68}). The requirements $\Omega_{3} = 0$ and $\Omega_{5}=0$, after substituting $u_1= 2 \phi _x $, lead to the two differential equations (\ref{eq106}) and (\ref{eq105}), respectively. Thus by making use of (\ref{eq105}) and (\ref{eq106}) we can write Eq. (\ref{eq70}) as $\Omega _{4} = \Omega_{3,x} + \phi _x \int \Omega_{5} dx \Rightarrow \mbox{(\ref{eq106})}_x + \phi _x \int \mbox{(\ref{eq106})} dx \Rightarrow 0$ and vanishes automatically. 
		Further, when we look for $\Omega_{6}=0$ from (\ref{eq73}), we can understand that $u_0$ can take any form in such a way that it satisfies the wave equation $u_{0,tt}- u_{0,zz}=0$. Here and in the forthcoming analysis, we consider a general solution $u_0(z,t)= f(t-z) + g(t+z)$ of the wave equation $u_{0,tt}- u_{0,zz}=0$ as the driving background and will study its implications with different localized nonlinear wave structure.}
		
		Hence, we can conclude here that when $\phi(x,y,z,t)$ satisfies Eqs. \eqref{eq105} and \eqref{eq106}, then Eq. (\ref{eq16}) holds and it is nothing but the required solution of the considered (3+1)D KPB equation (\ref{e1}). 
	\end{proof}
	
	\section{Nonlinear Waves with Controllable Background \& their Dynamics}\label{sec3}
	By using the generalized solution structure given by Eq. (\ref{eq16}), one can construct various nonlinear wave solutions on different type of physically significant backgrounds of interest. To validate such special provisions, in the following part, we have demonstrated a convenient way of obtaining one/first-order (i) solitary wave (soliton) and (ii) rational (rogue) wave solutions on interesting periodic as well as localized backgrounds incorporated through Jacobi elliptic functions and appropriate initial seed solutions. \textcolor{black}{Here, we wish to highlight that the higher-order soliton and rogue waves along with different other types of nonlinear wave solutions can be constructed by proceeding with this straightforward algorithm and their dynamics due to the influence of arbitrary driving backgrounds can be explored.} 
	\subsection{Kink Soliton on Periodic \& Localized Backgrounds}
	Taking the initial seed solution $\phi (x,y,z,t)$ in the form of a first order exponential function representing a solitary wave as below.
	\begin{equation}
		\phi (x,y,z,t) = 1 + \epsilon_1 \exp{(\alpha_1 x + \beta _1 y + \gamma_1 z + \delta _1 t  + \epsilon _0)}.  \label{eq91}
	\end{equation}
	Upon the substitution of (\ref{eq91}) into equations \eqref{eq105} and \eqref{eq106}, we get $\alpha_1 ^3 \beta _1 - \gamma _1 ^2 + \delta _1 ( \alpha_1 + \beta _1 + \delta _1) =0$, while solving it for $\beta_1$, we obtain $\beta _1 = {[\gamma_1 ^2 - ( \alpha_1 + \delta_1) \delta_1]}/{(\alpha_1^3 + \delta_1)}$. Thus the explicit form of $\phi (x,y,z,t)$ becomes
	\begin{equation}
		\phi (x,y,z,t) = 1 + \epsilon_1 \exp\left({\alpha_1 x + \frac{\gamma_1 ^2  - \delta_1 (\alpha_1 + \delta_1 )}{\alpha_1 ^3 + \delta_1 } y + \gamma_1 z+  \delta_1 t + \epsilon_0}\right).  \label{eq94}
	\end{equation}
	We wish to mention that for a special choice $\epsilon_1 = \epsilon, \epsilon_0 = 0, \alpha_1 =k , \gamma_1 =m, \delta_1 = -m \lambda_{10} $, one can get $\phi (x,y,z,t)$ and resulting solution as given in \cite{jpy}, obtained using bilinear B\"acklund transformation, but without the background functions that are going to play crucial role. 
	Hence from the general form of $\phi (x,y,z,t)$ from (\ref{eq94}), the required kink-soliton solution of KPB equation \eqref{e1} on spatio-temporal background is obtained as 
	\bea
	u(x,y,z,t) = \dfrac{2 \alpha_1 \epsilon_1 \exp{\left (\alpha_1 x + \frac{(\gamma_1 ^2 - ( \alpha_1 + \delta_1) \delta_1)y }{\alpha_1^3 + \delta_1}+\gamma_1 z + \delta_1 t+\epsilon_0\right )}}{1+\epsilon_1 \exp{\left (\alpha_1 x + \frac{(\gamma_1 ^2 - ( \alpha_1 + \delta_1) \delta_1)y }{\alpha_1^3 + \delta_1}+\gamma_1 z + \delta_1 t+\epsilon_0\right )}}+ f(t-z) + g (t+z),\label{eq96}
	\eea
	where $\alpha_1$, $\gamma_1$, $\delta_1$, $\epsilon_0$, and $\epsilon_1$ are arbitrary real parameters.  
	\begin{figure}[h]
		\centering
		\includegraphics[width=0.325\linewidth]{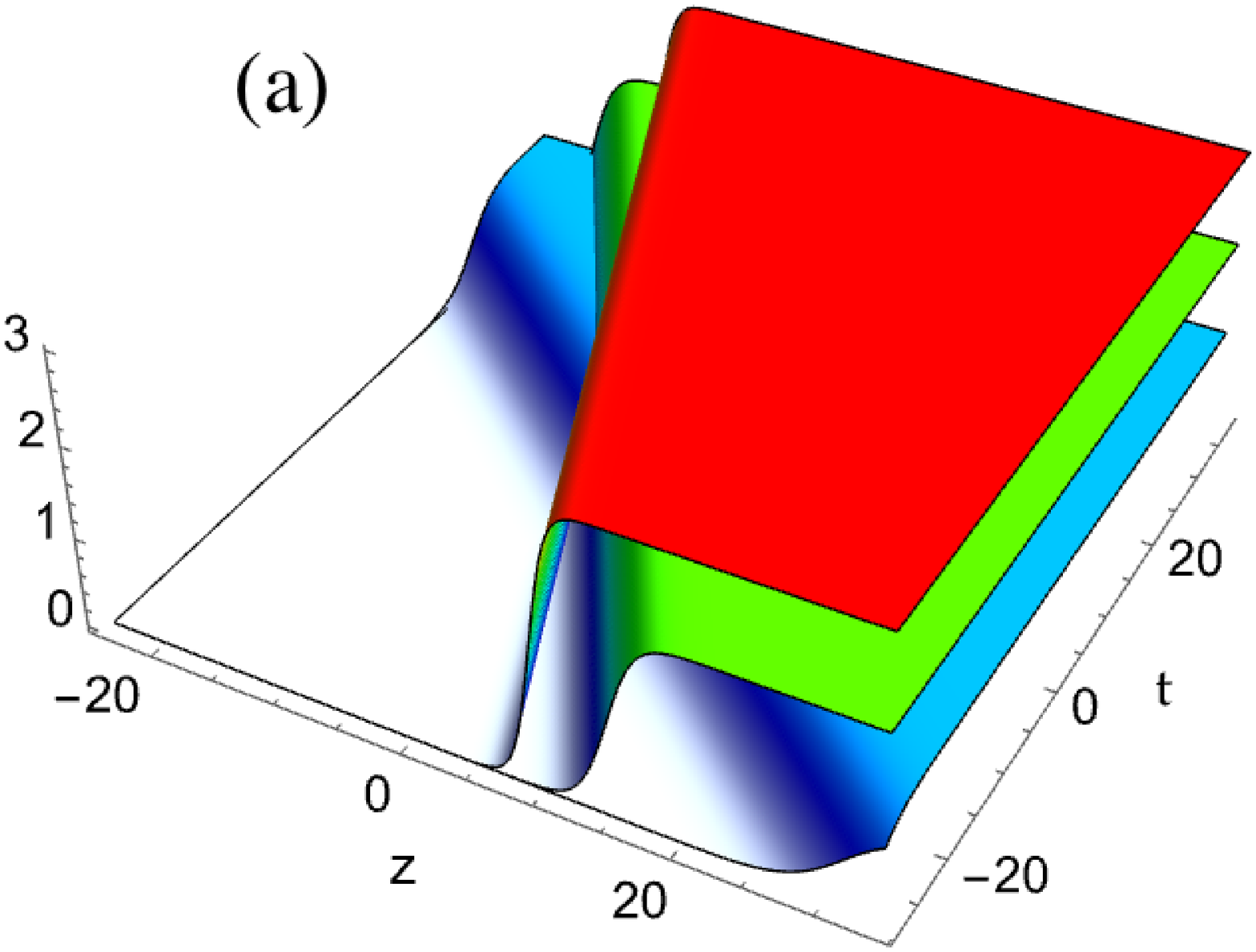}~~\includegraphics[width=0.325\linewidth]{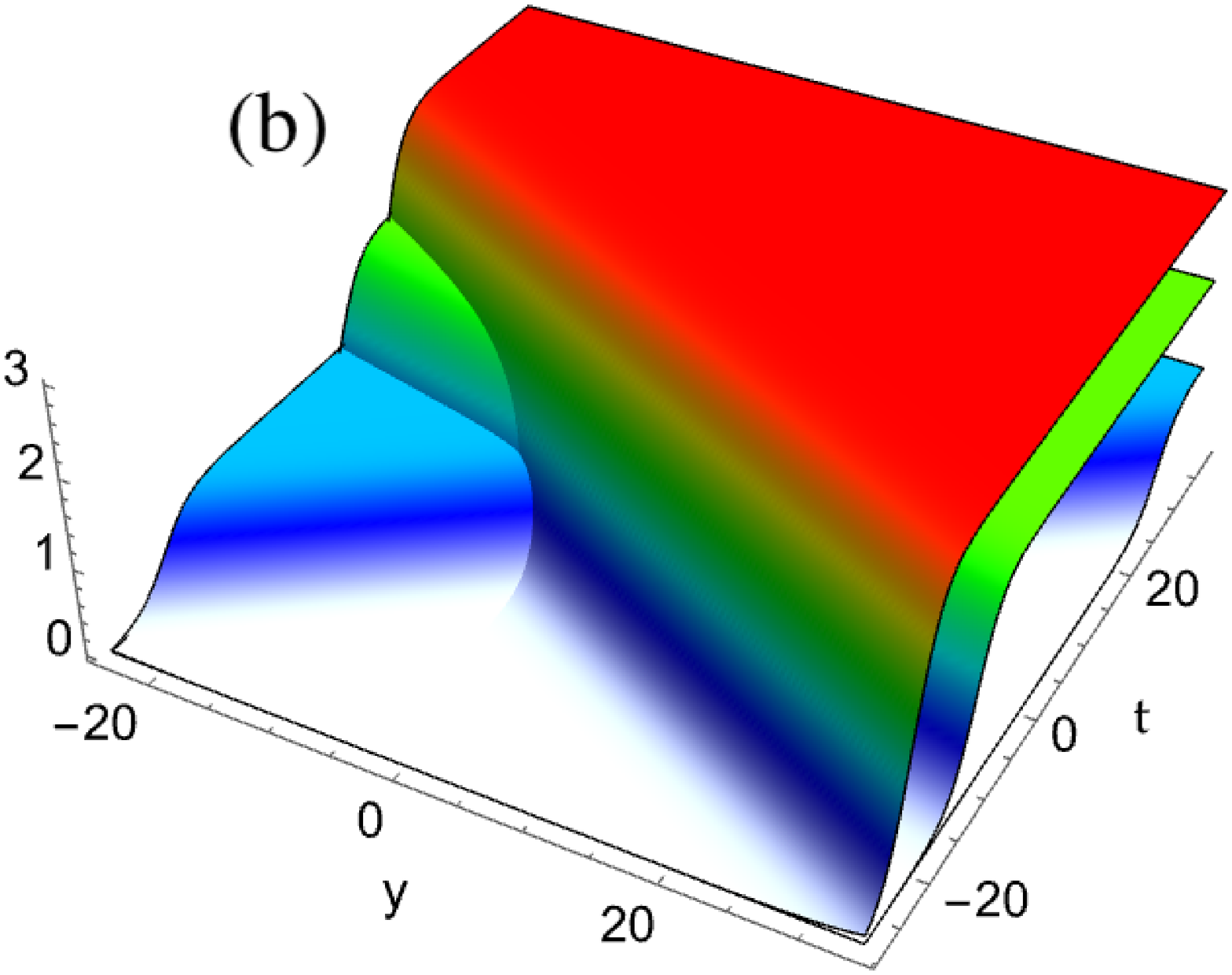}~~\\\includegraphics[width=0.325\linewidth]{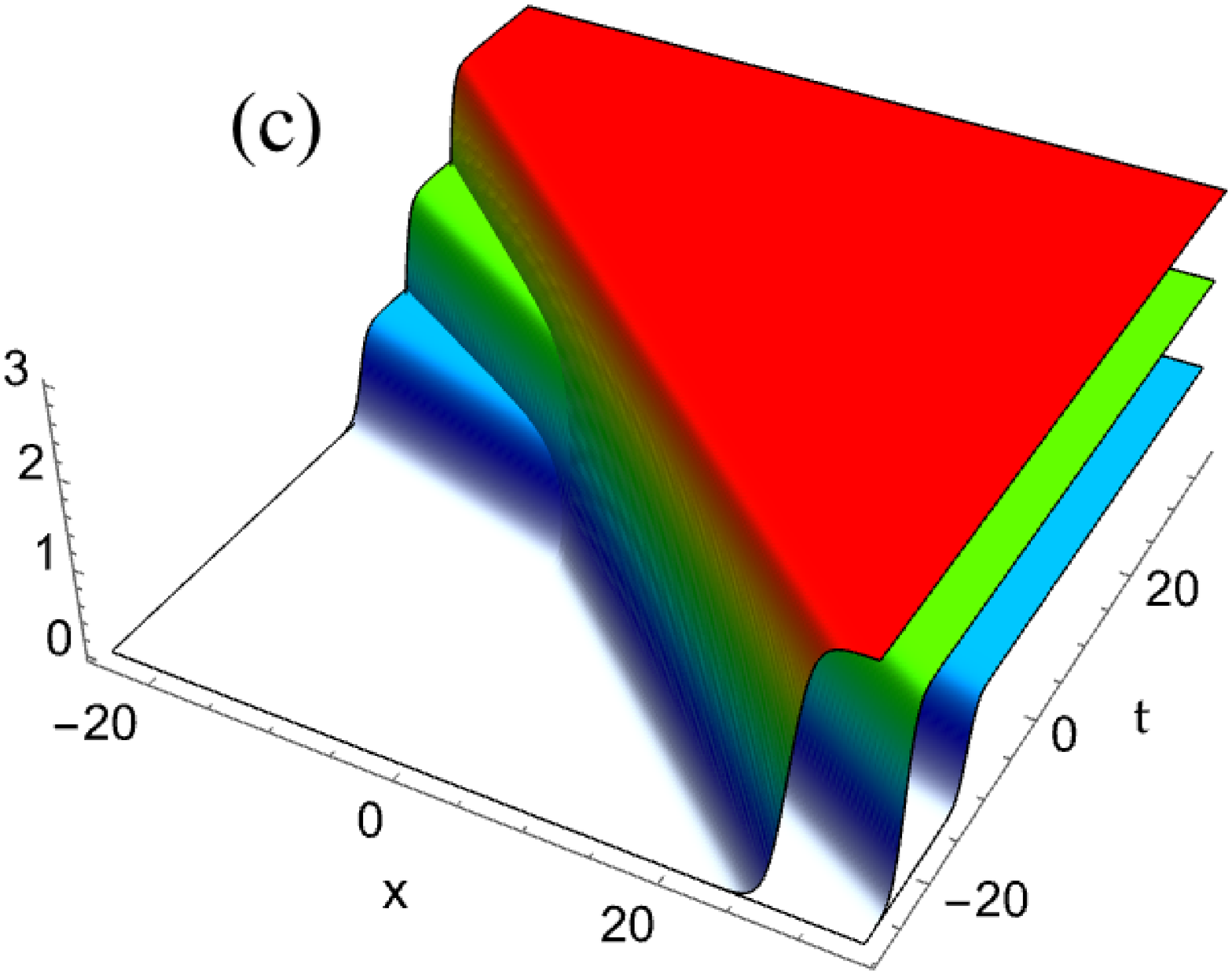}~~\includegraphics[width=0.325\linewidth]{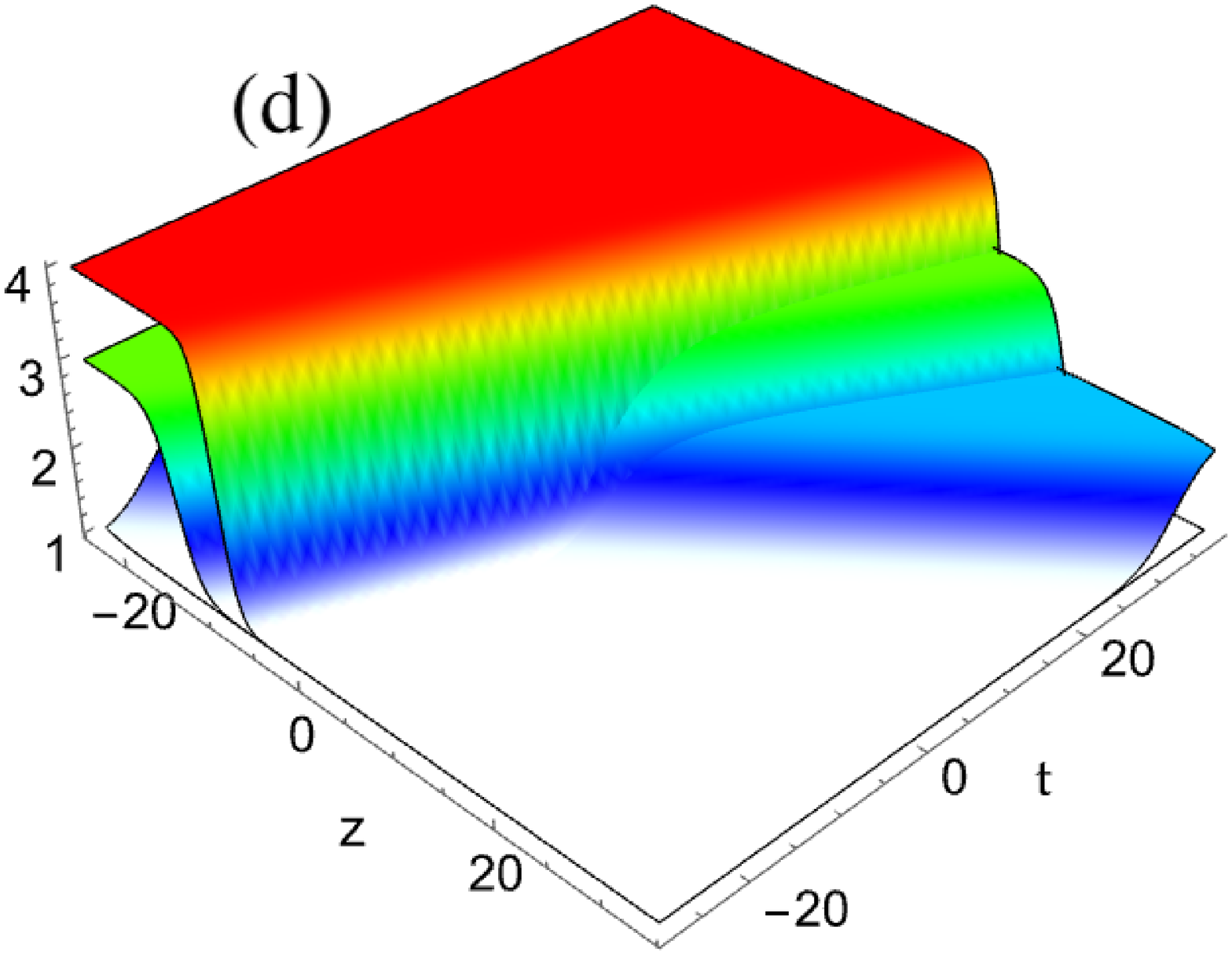}
		\caption{Dynamics of first-order (one) kink soliton $u$ given by Eq. (\ref{eq96}) without background along (a) $z-t$ at $x=y=0.4$, (b) $y-t$ at $x=z=0.4$, (c) $x-t$  at $z=y=0.4$ and (d) on a non-zero (constant) background along $z-t$ plane with $\delta_1=1.5$, $\epsilon_0=0.5$, and  $\epsilon_1=0.5$. For blue: $\alpha_1=\gamma_1=0.5$; green: $\alpha_1=\gamma_1=1.0$; and red: $\alpha_1=\gamma_1=1.5$.}
		\label{kink-zero}
	\end{figure}
	
	It is important to stress again that the sum of arbitrary functions $f(t-z)$ and $g (t+z)$ 
	can be used to control and modulate the dynamics of obtained nonlinear wave solution/structure. Here we can divide the above solution (\ref{eq96}) into two parts (i) without any background for $f=g=0$ and (ii) with controllable background when $f,g\neq 0$. In one hand, solution (\ref{eq96}) is nothing but a kink-type soliton (solitary wave) having a stable/smooth step-like profile without any background for vanishing functions $f=g=0$. Also, this kink soliton structure is localized in space and stably propagating in time, which can be referred to singly-localized or spatially-localized nonlinear wave. The arbitrary constant parameters $\alpha_1$, $\gamma_1$, $\delta_1$, $\epsilon_0$, and $\epsilon_1$ help to characterize some of its physical properties. Especially, the amplitude of the kink soliton is governed by $2 \alpha_1 \epsilon_1$, while its velocity becomes $-\delta_1/{\alpha_1}$, $-\delta_1({\alpha_1^3 + \delta_1})/{(\gamma_1^2-( \alpha_1 + \delta_1) \delta_1)}$, and  $-\delta_1/{\gamma_1}$ along different spatial dimensions $x$, $y$ and $z$, respectively. For an easy understanding, we have demonstrated the nature of kink soliton (\ref{eq96}) possessing different amplitude, velocity and phase in Fig. \ref{kink-zero} for various  choices of parameters as given in the caption. The obtained solution (\ref{eq96}) is only the first-order (single) kink soliton, while one can construct multi-kink solitons also in the standard way from appropriate initial/seed solution for $\phi(x,y,z,t)$. As the main objective of this work is to unearth the possibility of arbitrary backgrounds and their importance in the localized nonlinear waves, for which we restrict ourselves with one-kink soliton and explore them below. Further, the significance of backgrounds in multi-kink solitons and their collision dynamics are also of future interest. 
	
	As already mentioned in the introduction, the possibility of constructing localized nonlinear wave solution on arbitrary backgrounds is more challenging and not much explored for higher-dimensional models. In order to showcase the influence of arbitrary backgrounds in our KPB model (\ref{e1}) one can choose any arbitrary form for $f(t-z)$ and $g(t+z)$, subject to satisfying the wave equation $u_{0,tt}- u_{0,zz}=0$. Here we adapt the form of $f$ and $g$ in terms of physically important Jacobi elliptic functions as below.
	\bes\bea 
	f(t-z)=a_1 \mbox{sn}(a_2(t-z),k)+a_3 \mbox{cn}(a_4(t-z),m)+a_5 \mbox{dn}(a_6(t-z),n),\\
	g(t+z)=b_1 \mbox{sn}(b_2(t+z),k)+b_3 \mbox{cn}(b_4(t+z),m)+b_5 \mbox{dn}(b_6(t+z),n),
	\eea \label{jacobi}\ees
	where $k$, $m$ and $n$ are the elliptic modulus parameters ($0\leq k,m,n \leq 1)$, while $a_j$ and $b_j$ ($j=1,2,3\dots,6$) are arbitrary real constants. Particularly, $a_1$, $a_3$, $a_5$, $b_1$, $b_3$ and $b_5$ indicate the amplitude/magnitude of the modulation, while $a_2$, $a_4$, $a_6$, $b_2$, $b_4$ and $b_6$ represent periodicity/localization of the background. As it is known that these Jacobi functions trigger periodic, localized, and combined (periodic+localized) variations for appropriately chosen modulus parameters, here we discuss how they affect the dynamics of kink soliton. 
	\begin{figure}[h]
		\centering
		\includegraphics[width=0.325\linewidth]{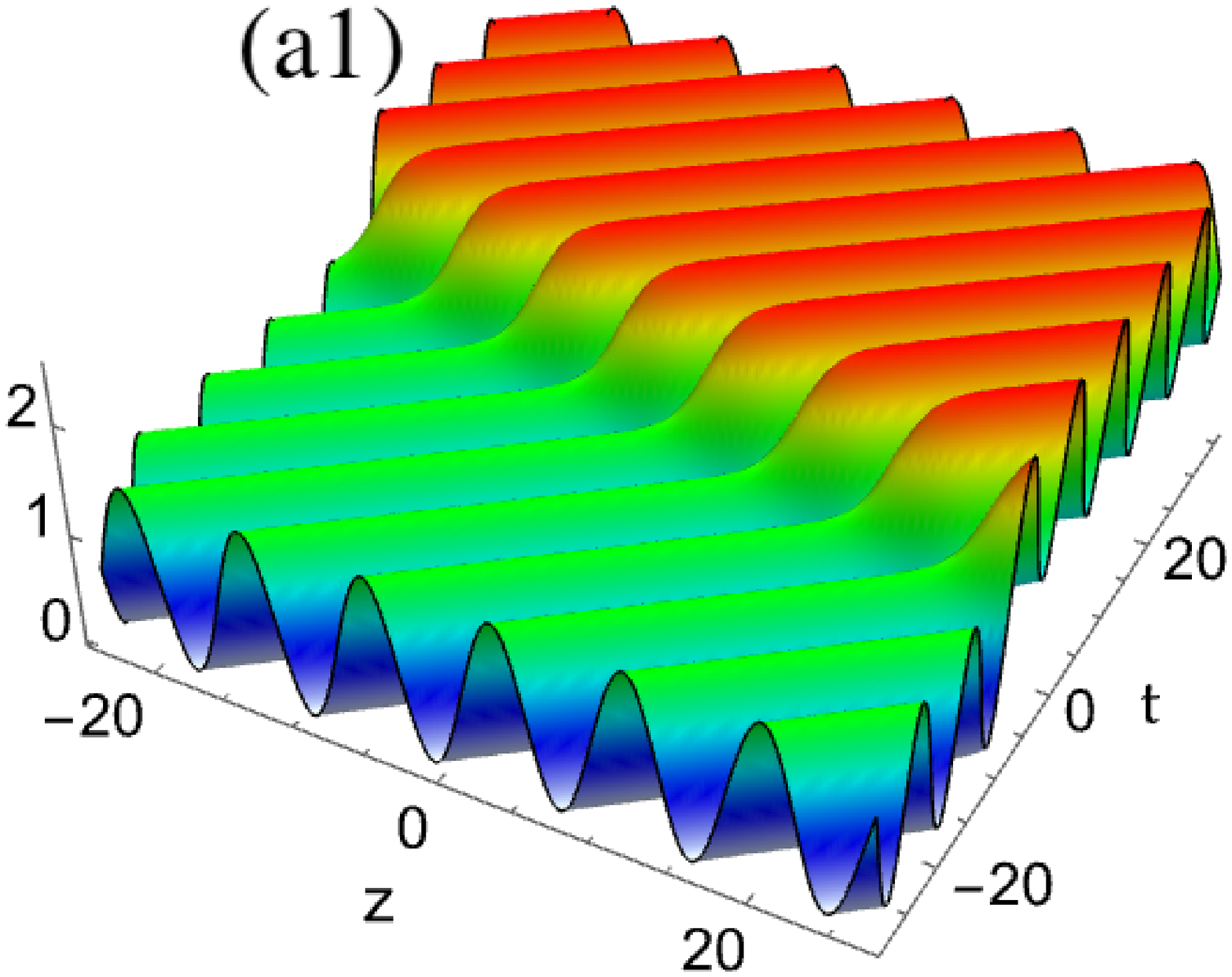}~
		\includegraphics[width=0.325\linewidth]{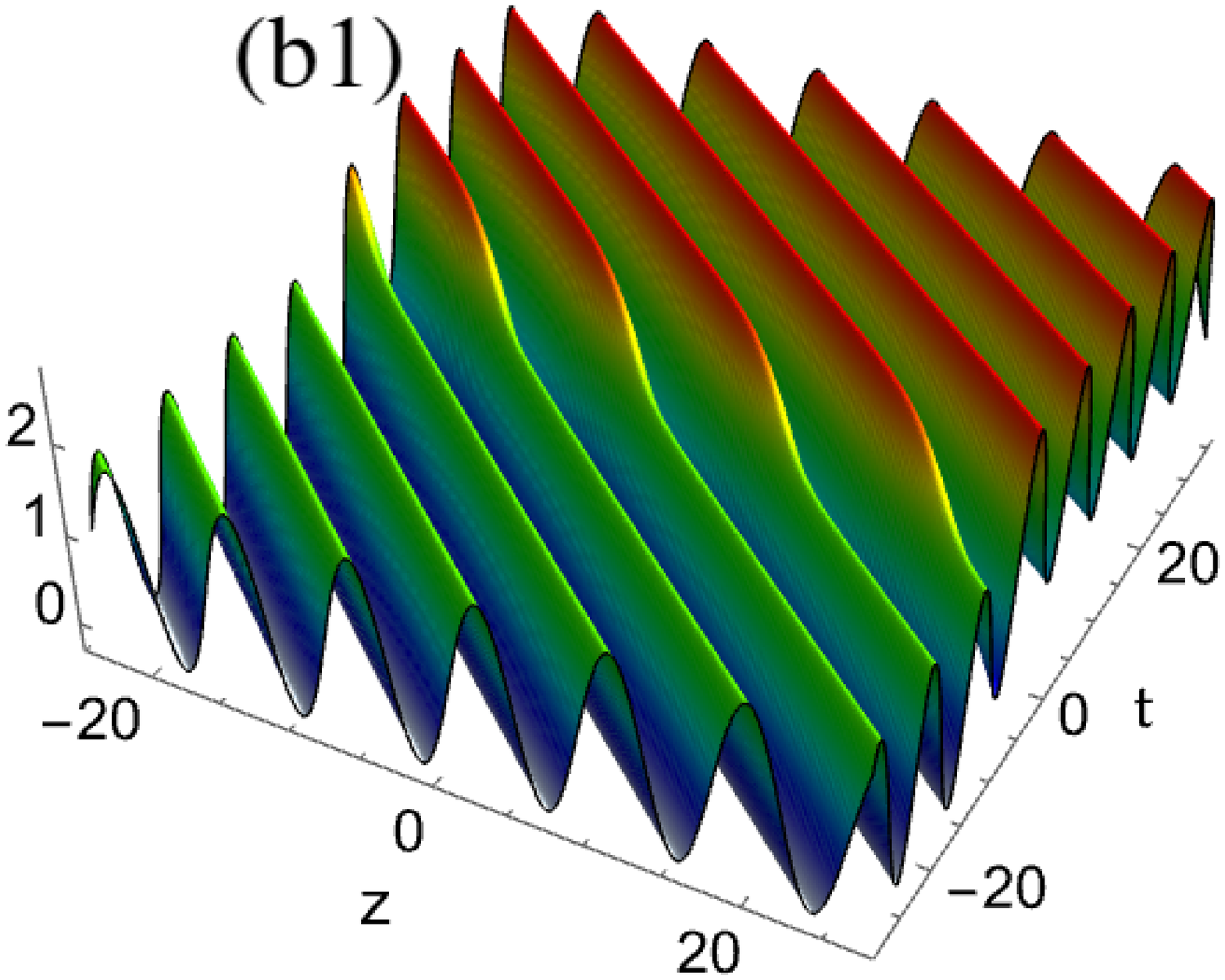}~\includegraphics[width=0.325\linewidth]{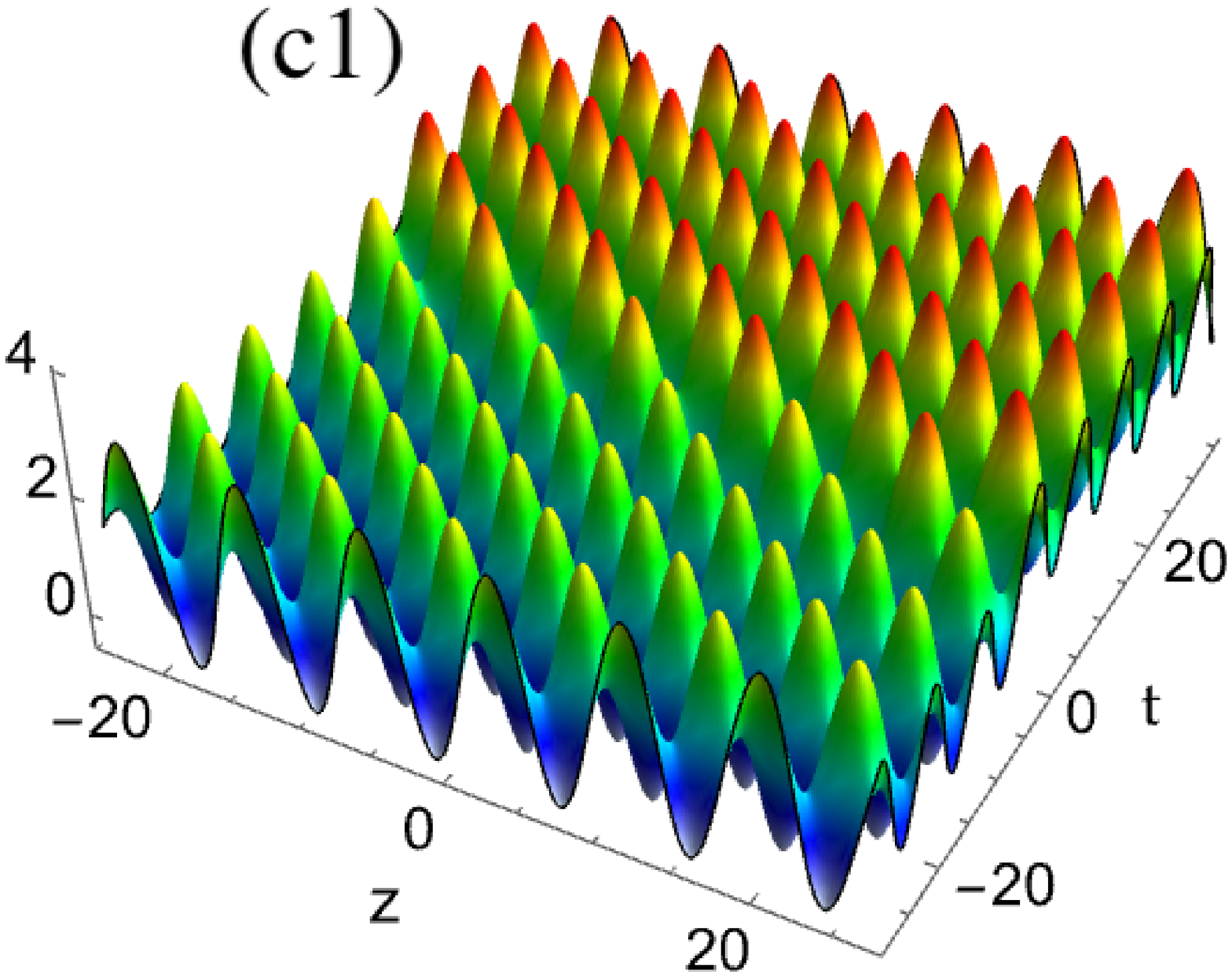}~\\
		\includegraphics[width=0.325\linewidth]{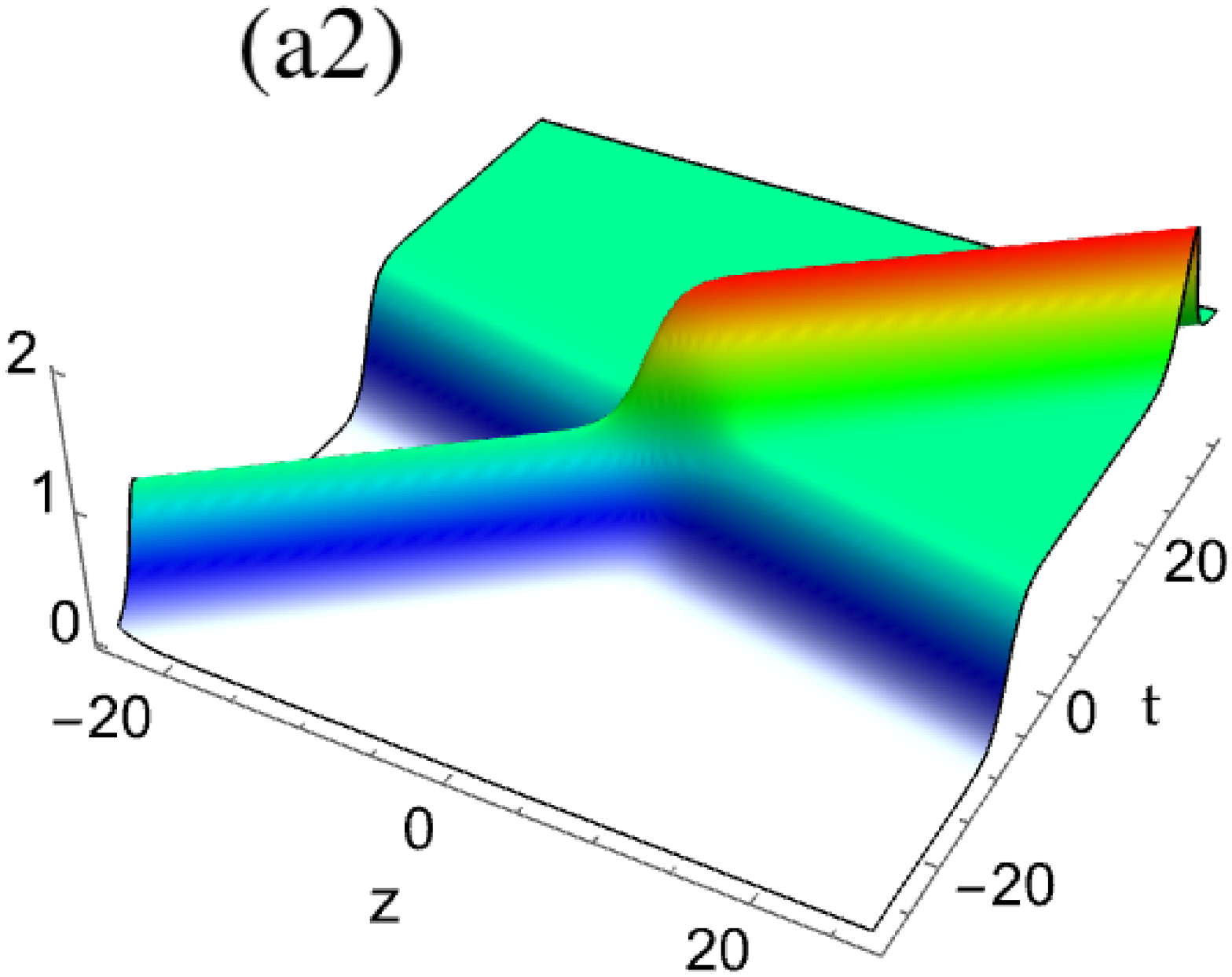}~
		\includegraphics[width=0.325\linewidth]{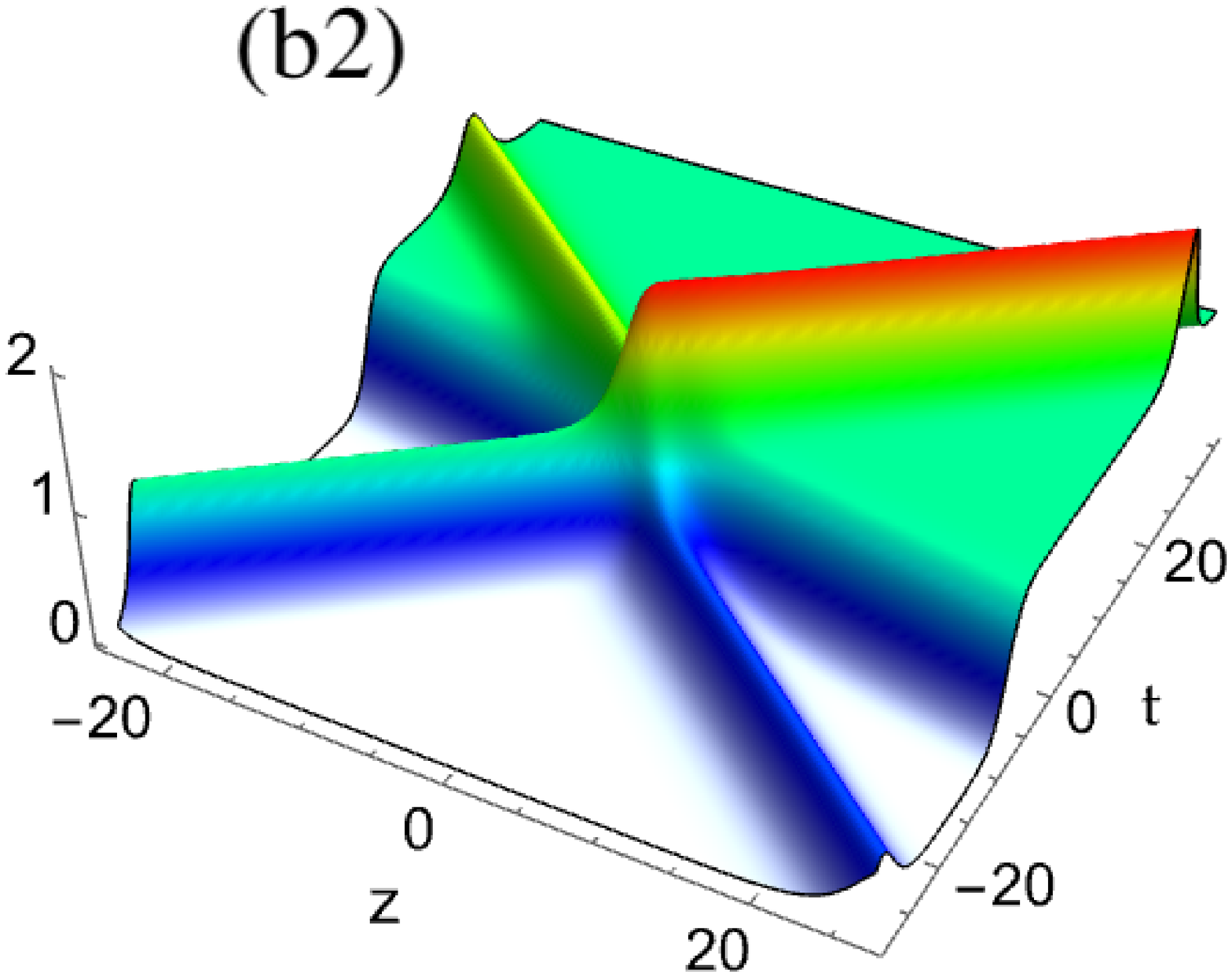}~\includegraphics[width=0.325\linewidth]{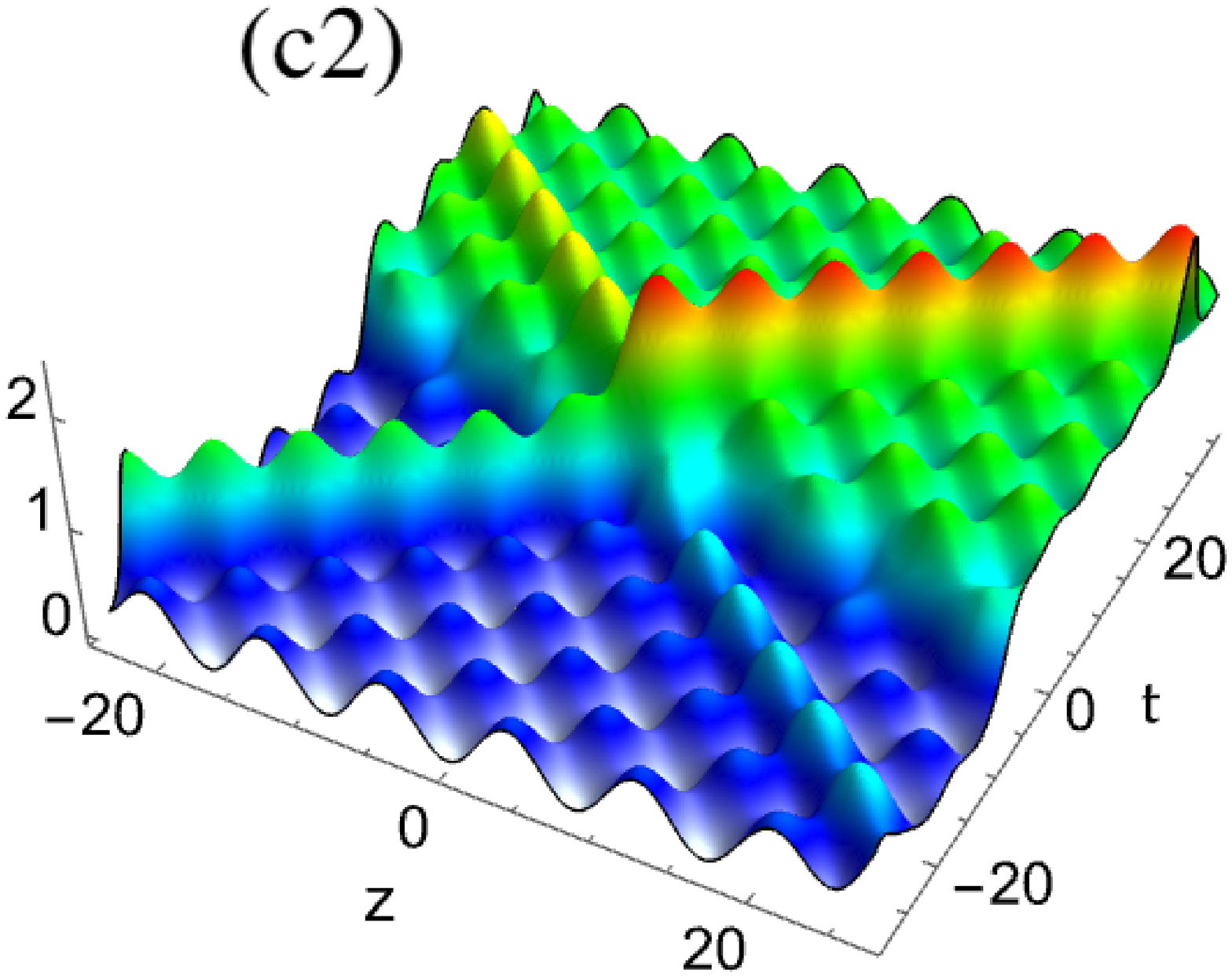}\\
		\includegraphics[width=0.325\linewidth]{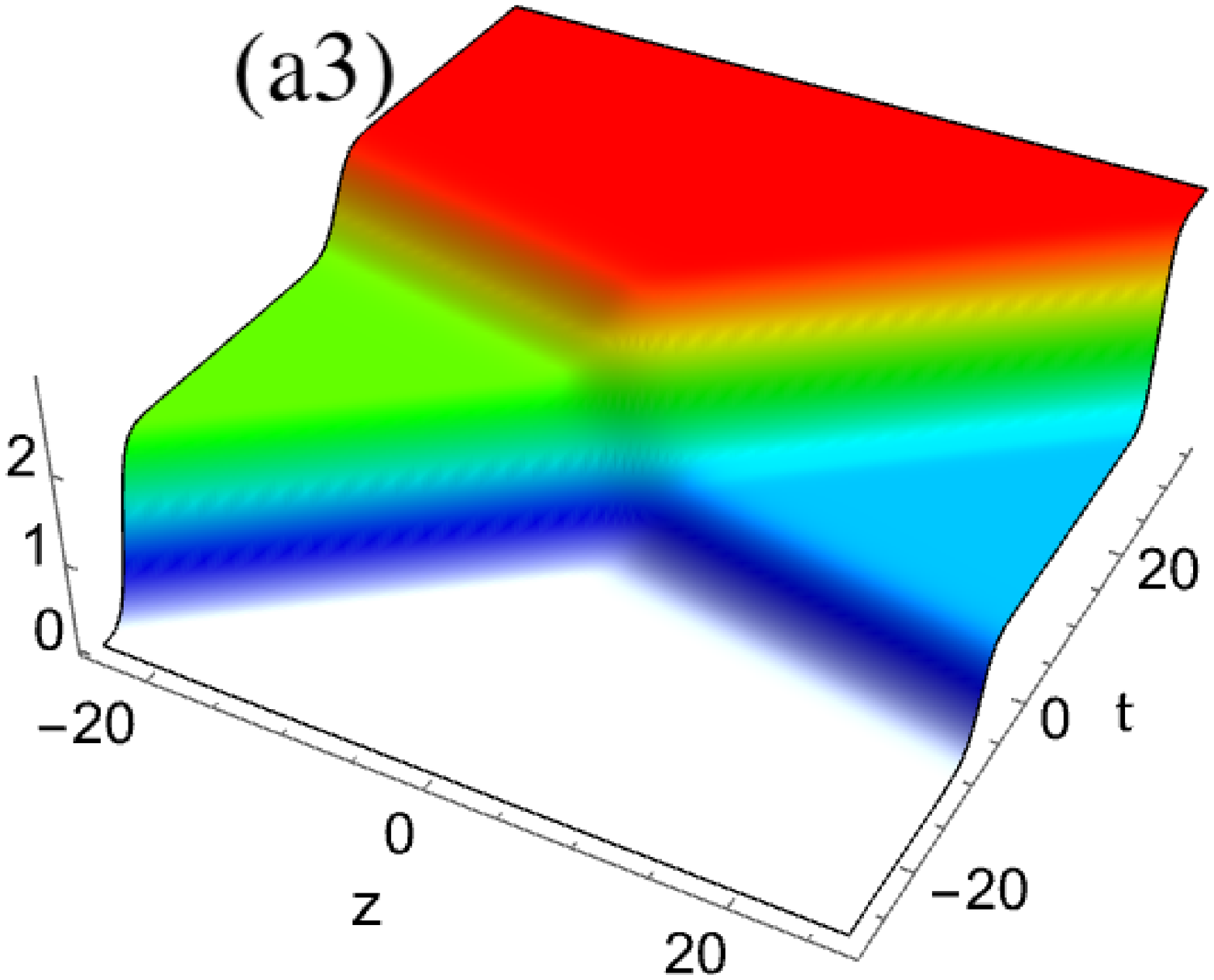}~
		\includegraphics[width=0.325\linewidth]{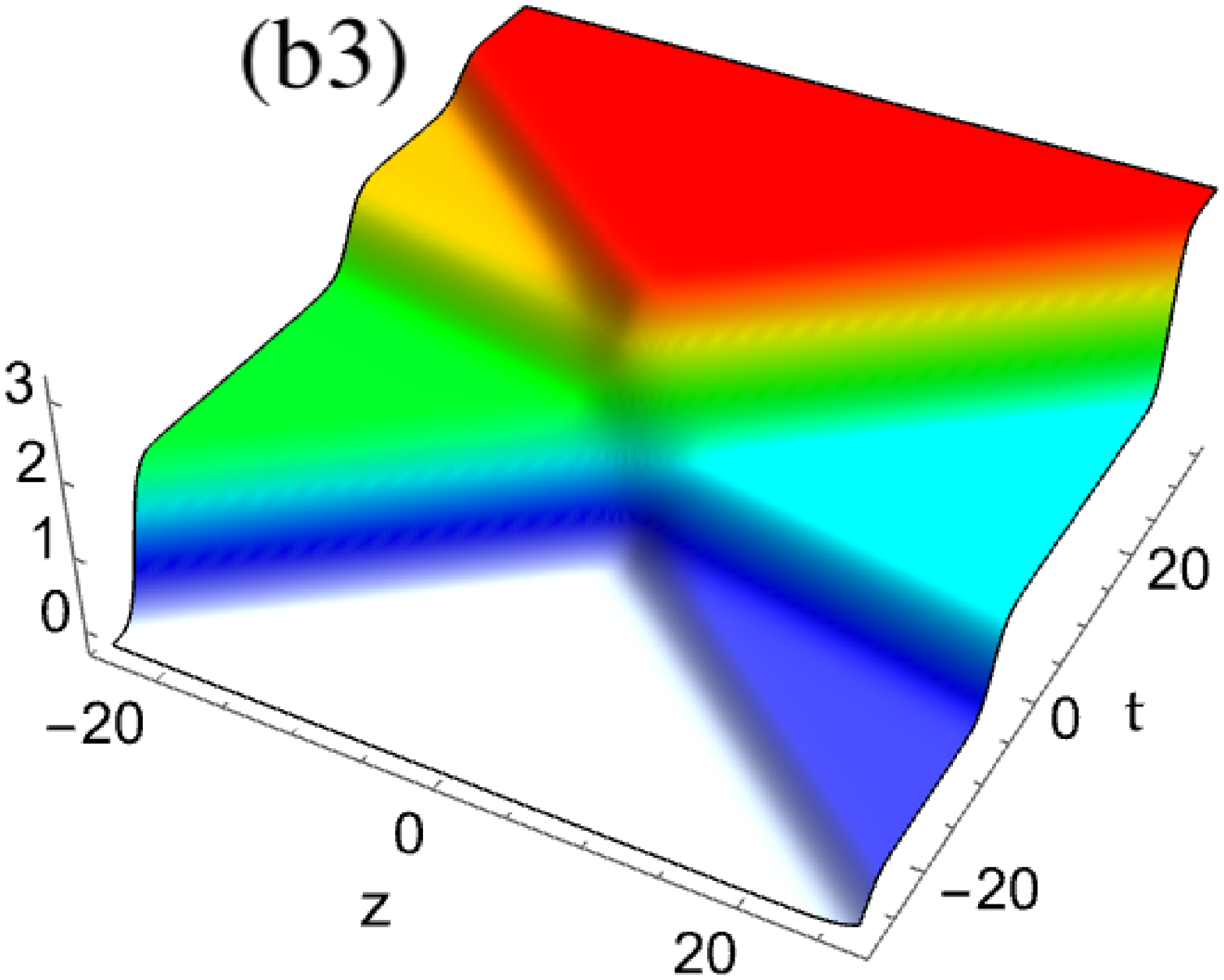}~\includegraphics[width=0.325\linewidth]{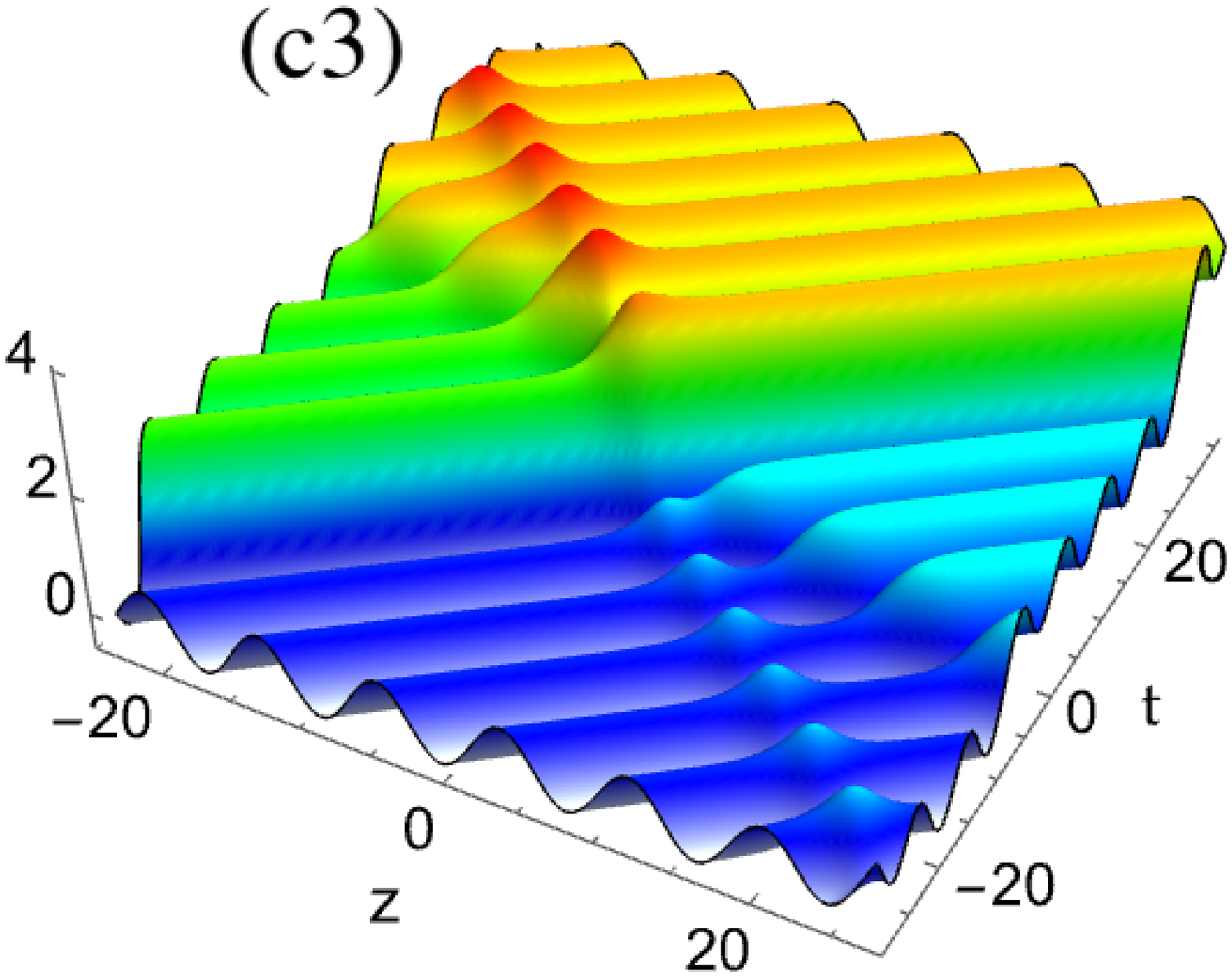}
		\caption{Impact of various backgrounds (\ref{jacobi}) on the dynamics of kink soliton $u$ given by (\ref{eq96}) along $z-t$ plane at $x,y=0.4$ with $\alpha_1=0.5$, $\gamma_1=0.25$, $\delta_1=1.0$, $\epsilon_1=0.5$, and  $\epsilon_0=0.25$. The background parameters are considered as (a1) periodic $a_1,a_2=0.75$, (b1) periodic $b_3=1.0$, $b_4=0.75$, and (c1) double-periodic $a_1,a_2,b_4=0.75$, $b_3=1.0$ with $k,m=0$; (a2) bell-kink-type solitons collision $a_3=1.0$, $a_4=0.75$, (b2) double-bell-kink solitons collision $a_3=1.0, b_3=0.5$, $a_4,b_4=0.75$, and (c2) bell-type localized and double-periodic $a_1,b_1=0.15,a_3=1.0,b_3=0.5$, $a_2,a_4,b_2,b_4=0.75$ with $k=0,m=1$; (a3) two kink-solitons collision $a_1=1.0$, $a_2=0.75$, (b3) three kink-solitons collision $a_1,b_1=0.5$, $a_2,b_2=0.75$, and (c3) double-bell-kink solitons collision with periodic excitation $a_1=1.0$, $a_3,b_3=0.5$,  $a_2,a_4,b_4=0.75$ with $k=1,m=0$, while all the other values are fixed as zero.}
		\label{kink-peri}
	\end{figure}
	
	First, we investigate kink soliton (\ref{eq96}) on single-periodic and double-periodic backgrounds with $a_1\neq 0$ or $b_3\neq 0$ and both $a_1,b_3\neq 0$, respectively, for elliptic modulus parameters $k,m=0$ with other $a_j$ and $b_j$ parameters as zero. We have depicted the resulting dynamical pattern as graphical outcome in Fig. \ref{kink-peri}(top panel). We can observe that the kink soliton still maintain its single-step structure, but it exhibits periodic and double-periodic amplitude excitation instead of smooth/flat propagation. The double-periodic nature results from two different spatio-temporal functions $f(t-z)$ and $g(t+z)$. Additionally, we can observe the quasi-double-periodic evolution of kink soliton by taking two different periods of oscillation $a_2\neq b_4$ with $a_1,b_3\neq 0$. One can also identify similar periodic and double-periodic kink soliton with $a_3\neq 0$ or $b_1\neq 0$ and both $a_3,b_1\neq 0$, respectively, with $k,m=0$ and $a_1,a_5,b_3,b_5=0$. Next, we consider localized wave-type background with elliptic moduli $m,n=1$ with either one or more among $a_3,b_3,a_5$ and $b_5$ as non-zero, which provides one or more localized bell-type (sech) solitary wave background. Thus the resultant wave dynamics here will be the superposition of initial kink soliton and background soliton (interacting sech solitons) as shown in the middle panel of Figs. \ref{kink-peri}(a2) and (b2). Interestingly, the background can also be chosen with double-localized and double-periodic nature which reveals much interesting dynamics as given in Fig. \ref{kink-peri}(c2) with $k=0$ and $m=1$ (or $n=1$). Moving further, when we choose kink-type background with $k=1$ and $a_1\neq 0$ or/and $b_1\neq 0$, the dynamics becomes more interesting as this case results into a collision of two or three kink-solitons as shown in Fig. \ref{kink-peri}(bottom panel). It is obvious that these solitons exhibit phase-shift after collisions and it is clearly visible from the demonstrations. Similar to the double-periodic case, one can also opt for double-localized wave backgrounds like sech+sech, sech+tanh, and tanh+tanh, along with additional (single or double) periodic wave too, and to elucidate this feature a few options are demonstrated in Figs. \ref{kink-peri}(c2) and (c3). Moreover, any other possible combinations/forms of background can be adopted and thereby the associated evolution can be explored extensively.
	
	\subsection{Rational (Rogue) Wave on Varying Periodic \& Localized Backgrounds}
	In the previous part, we have studied the dynamics of kink soliton on various backgrounds. Proceeding further, in this part, we analyse another interesting nonlinear wave entity, which is nothing but a doubly-localized rational (rogue) wave and its implications due different spatio-temporally varying backgrounds. For this purpose, we take another type of very simple polynomial form of initial seed solution $\phi (x,y,z,t)$ defined as 
	\begin{equation}
		\phi= \alpha_0+\alpha_1 x + \alpha_2 y + \alpha_3 z + \alpha_4 t + \beta_1 x y + \beta_2 xz + \beta_3 xt + \beta_4 yz + \beta_5 yt+ \beta_6 zt + \gamma_1 x^2 + \gamma_2 y^2 + \gamma_3 z^2 + \gamma_4 t^2. \label{eq97} 
	\end{equation}
	Substituting the above $\phi$ into (linear and nonlinear) differential equations \eqref{eq105} and \eqref{eq106}, we get a set of relations among the coefficients $\alpha_j$, $\beta_j$, and $\gamma_j$ that provide the exact form of $\phi$ as\vspace{-0.35cm}
	\bes\bea 
	\phi(x,y,z,t) = \alpha_0 +\dfrac{\Lambda_1}{\beta_3^2 \delta_2 \delta_3}, \label{peq1} \vspace{-0.35cm}
	\eea
	where
	\bea 
	&\Lambda_1 =& \sqrt{6} (t-y) \beta_1 \beta_3^3 \gamma_1 \delta_2+(-\alpha_3 \beta_3^{3/2}(\beta_3 t+\delta_1 y)+(\beta_3^2 (\alpha_1+\beta_1 y+\beta_3 t)x+ \beta_3^2\gamma_1 x^2 \nonumber\\ 
	&&+\beta_3^2( \alpha_3 z+ \gamma_4 t^2)-(\alpha_1 \beta_3^2+ \beta_3^3 t+4  \sqrt{\beta_3} \gamma_1 \delta_2 \gamma_4 z -2 \beta_3 \gamma_4 \delta_1 t)y+y^2(\gamma_4 \delta_1^2-\beta_3^2(\beta_1+\gamma_1)))\delta_2 \nonumber\\ 
	&&-z\sqrt{\beta_3}((x-y)\beta_3^2+2 (\beta_1 y+\beta_3 t)\gamma_4)\delta_2^2+z^2\beta_3 \gamma_4 \delta_2^3)\delta_3+\sqrt{6}\beta_3 \delta_2 \delta_3^2 y,\\
	&&\delta_1=\beta_1+2\gamma_1, \quad  \delta_2=\sqrt{\beta_1+2\gamma_1+\beta_3}, \quad \delta_3=\sqrt{-\beta_1 \beta_3  \gamma_1  \delta_2^2}. \label{peq2}
	\eea 
	Thus, from the above explicit form of $\phi(x,y,z,t)$, we can deduce the required rational wave solution of KPB equation \eqref{e1} from (\ref{eq16}) as below.
	\bea
	u(x,y,z,t) = \dfrac{2 \phi_x}{\phi} + f(t-z)  + g (t+z)\Rightarrow \dfrac{2 \left( \alpha_0 \delta_2 \delta_3\beta_3^2+\Lambda_{1,x}\right)}{\alpha_0 \delta_2 \delta_3\beta_3^2+\Lambda_{1}} + f(t-z)  + g (t+z). \label{eq910}
	\eea \ees 
	The above solution (\ref{eq910}) is a simplest form of rational soliton or rogue wave in both spatio-temporally (doubly) localized excitation and amplitude peaks with extended tails on either side. Here we have different arbitrary parameters such as $\alpha_0$, $\alpha_1$, $\alpha_3$, $\beta_1$, $\beta_3$, $\gamma_1$, and $\gamma_4$ to manipulate its identities like symmetric or asymmetric type rogue wave structure with controllable peak amplitude, width, and tail length/orientation. For a proper understanding on the significance of these parameters, we have shown the rogue wave patterns appearing along $z-t$ plane for different choices in Fig. \ref{rogue-const} where Fig. \ref{rogue-const}(a) and (b) represent a symmetric rogue wave, while \ref{rogue-const}(c) and (d) denote the rogue waves with induced asymmetry in its amplitude and width. Further, one can also observe such rogue waves in other two ($x-t$ and $y-t$) dimensions too.
	\begin{figure}[h]
		\centering
		\includegraphics[width=0.325\linewidth]{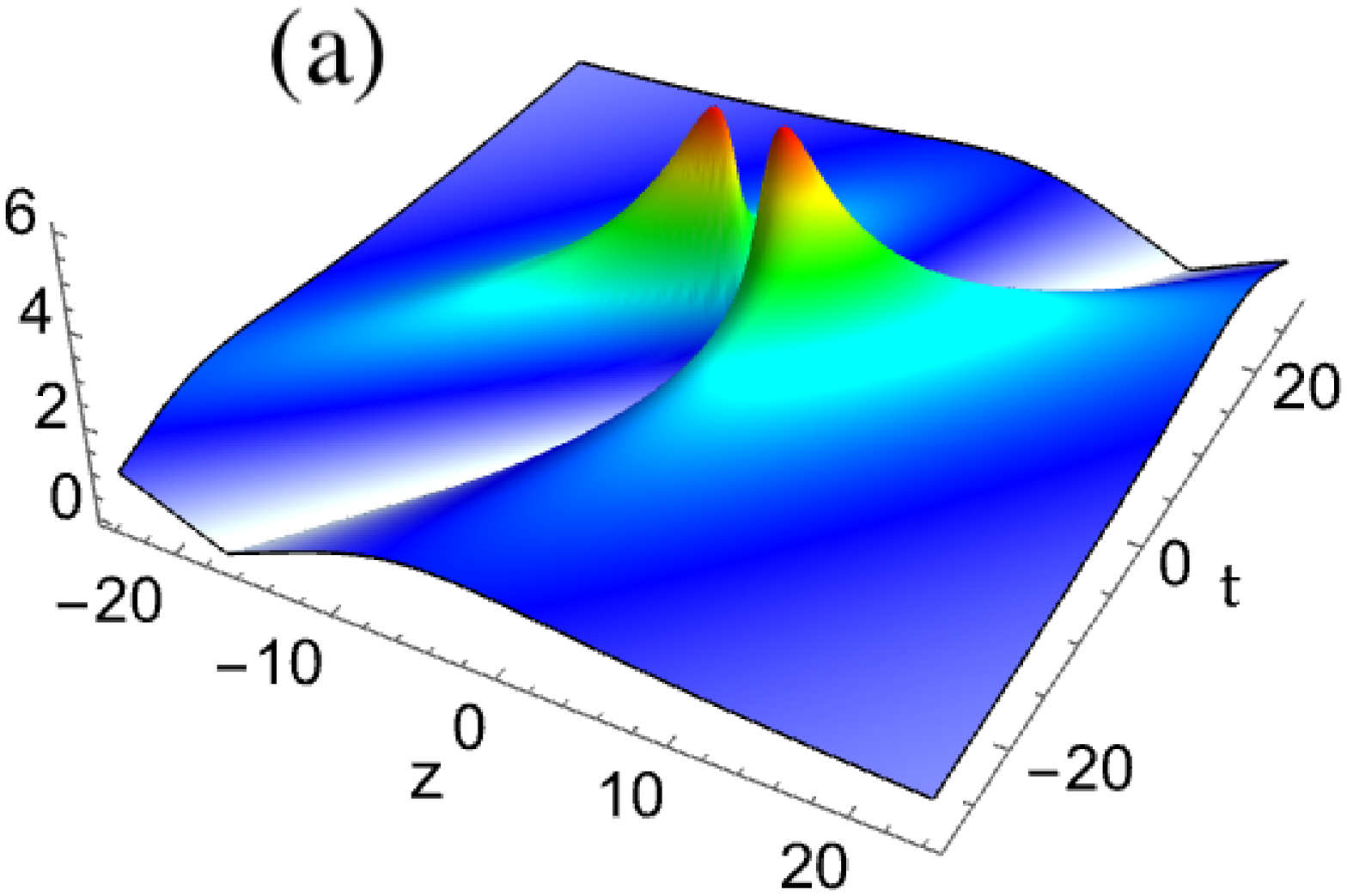}~~\includegraphics[width=0.260\linewidth]{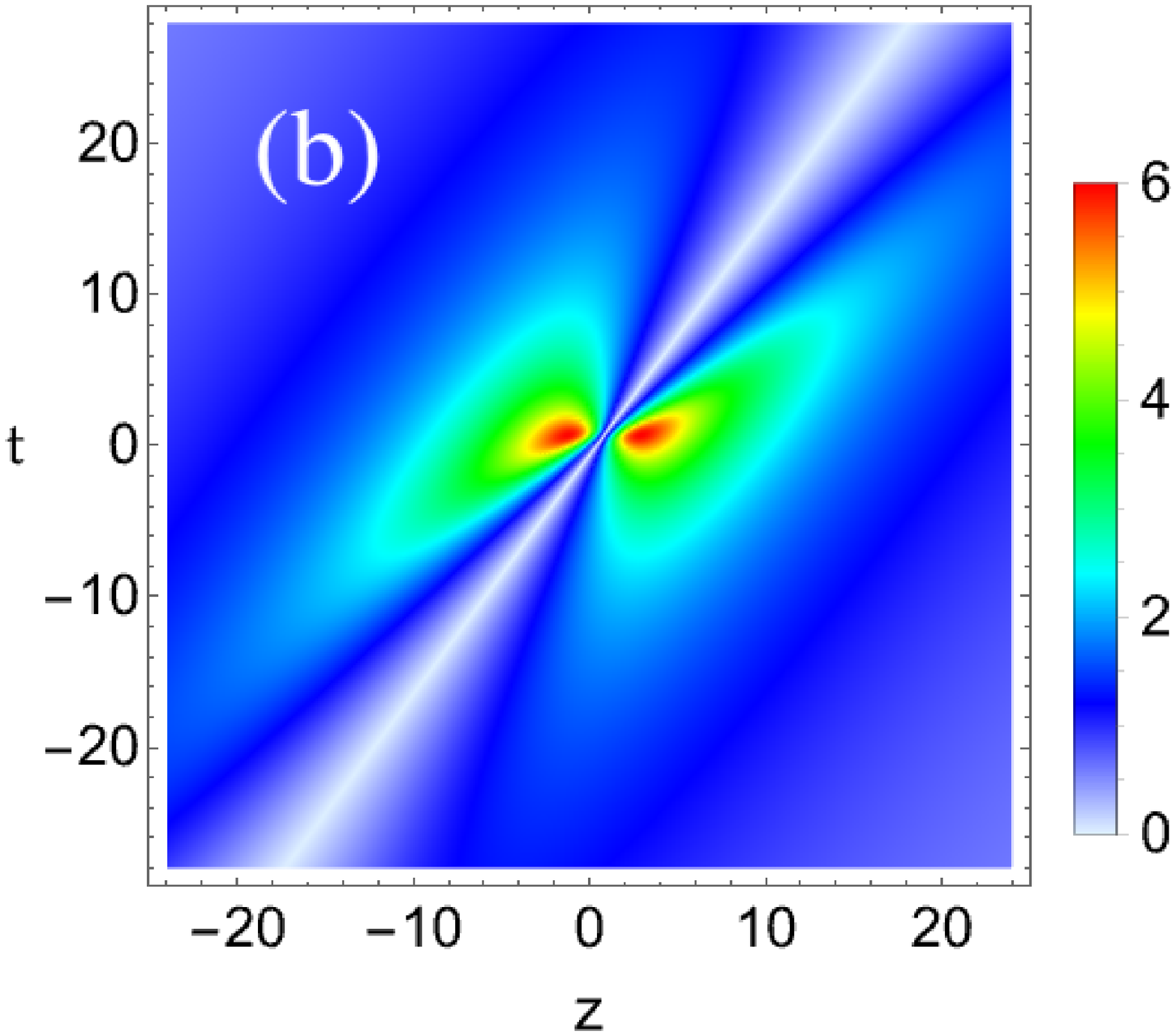}~~\\\includegraphics[width=0.325\linewidth]{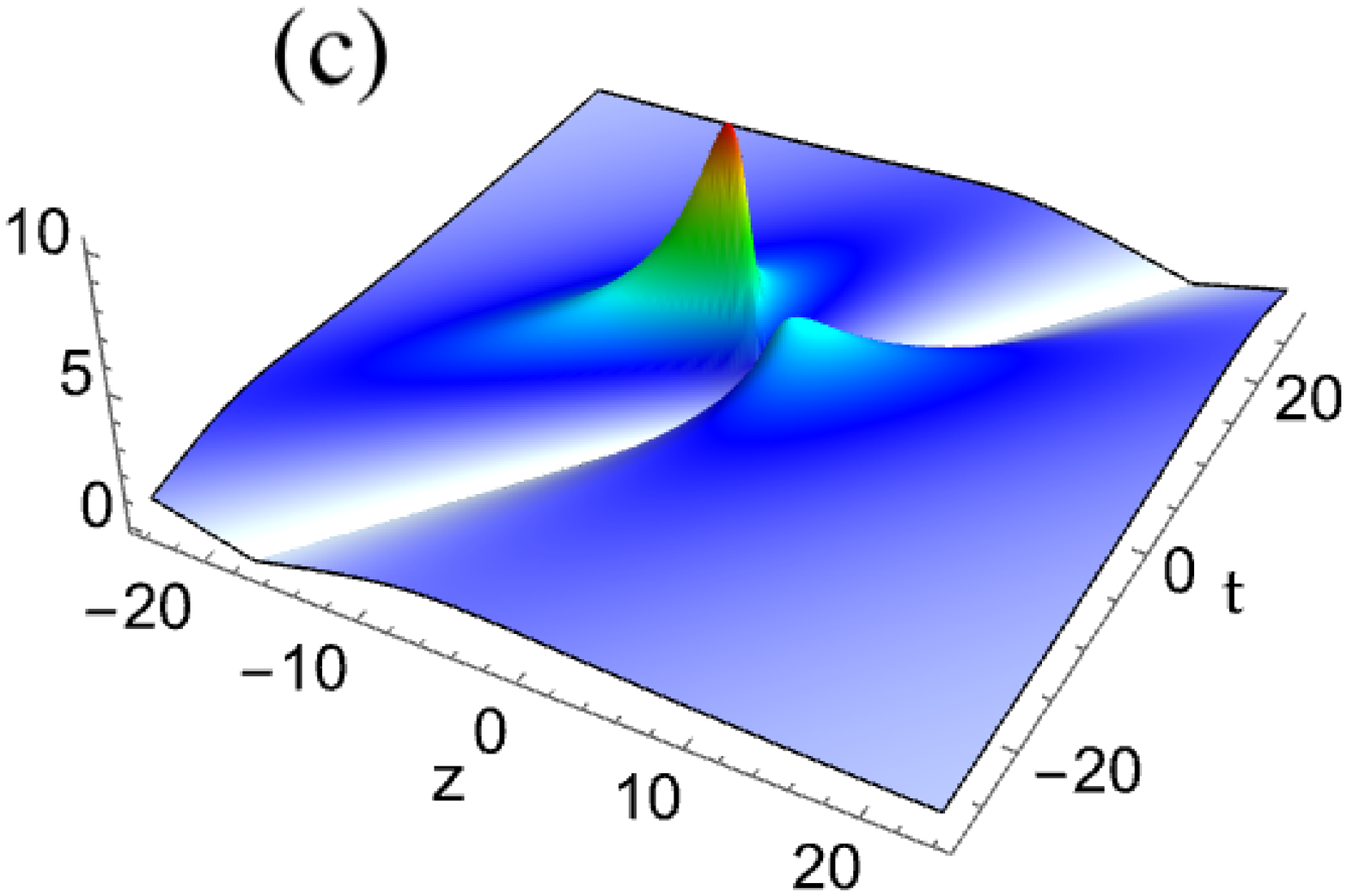}~~\includegraphics[width=0.325\linewidth]{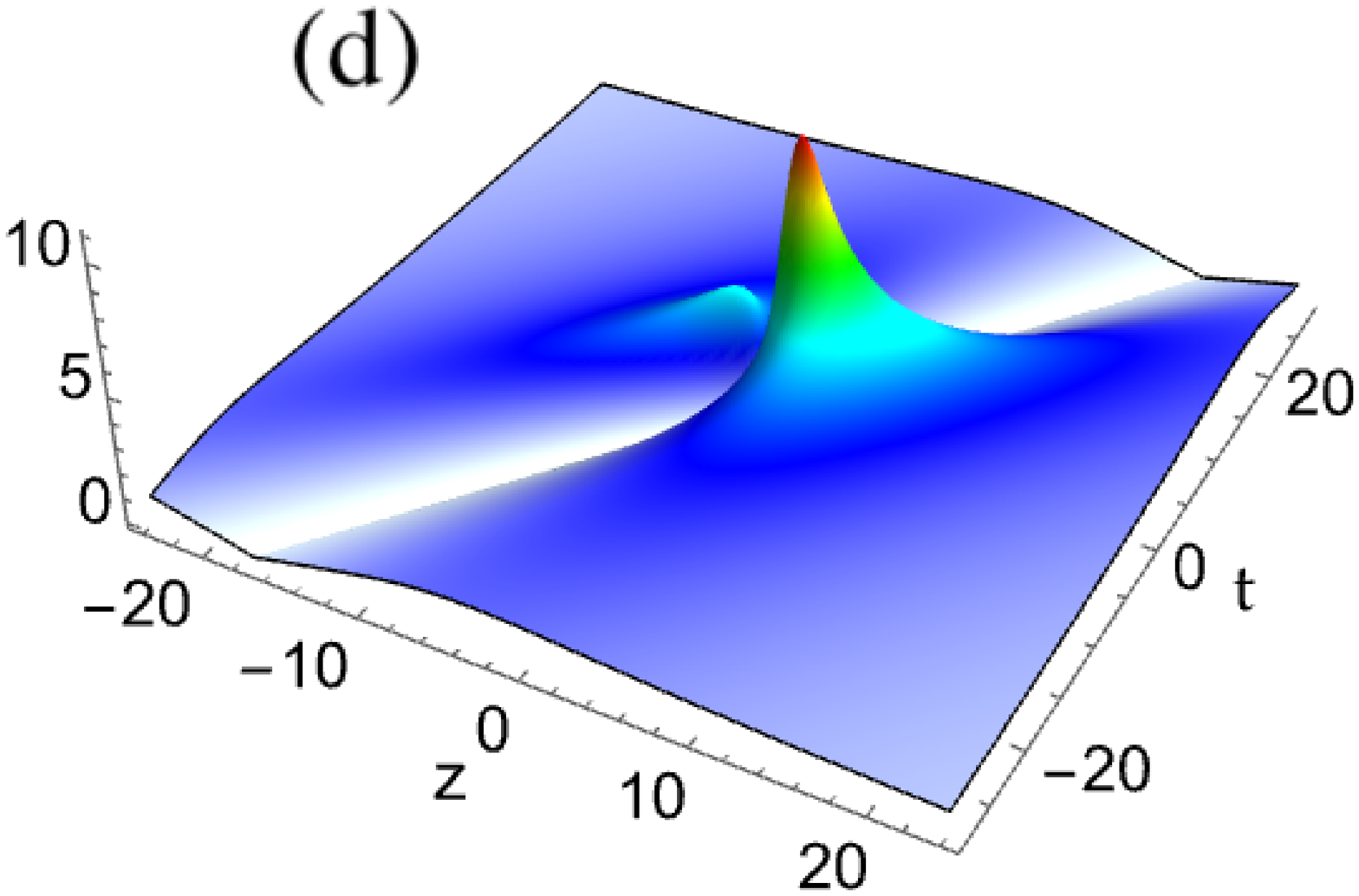}
		\caption{Dynamics of (a,b) symmetric and (c,d) asymmetric peak rational (rogue) wave $|u|$ given by (\ref{eq910}) on stable ($f=g=0$) background in $z-t$ plane at $x,y=0.4$ for $\alpha_0=0.5$, $\alpha_1=0.15$, $\beta_1=0.5$, $\beta_3=1.25$, $\gamma_1=0.5$, $\gamma_3=0.5$, and $\gamma_4=0.05$ with (a,b) $\alpha_3=0.26$, (c) $\alpha_3=0.56$, and (d) $\alpha_3=0.06$.}
		\label{rogue-const}
	\end{figure}
	
	When we opt for background to take any arbitrarily varying spatio-temporal form from several combinations of Jacobi elliptic functions (\ref{jacobi}), it influences the rogue wave to undergo significant changes in its overall dynamics. 
	For a constant background ($f,g=0$) it is smooth, whereas for varying $f(t-z)$ and $g(t+z)$ arising from (\ref{jacobi}) with $k,m=0$ we can witness that the rogue wave escalate onto the single- and double-periodic backgrounds with modulation in its amplitude, width and tail characteristics. Important identities that execute for background modulation are its amplitude which increases with suppression in the width and its tails leading to an elongation or suppression as depicted in the top panel of Fig. \ref{rogue-modu}. Similar to the kink soliton case, here also the single or double localized wave(s) background (for elliptic moduli $m,n=1$) and combined periodic+localized (for $k=0$ and $m,n=1$) wave background reveal the persisting rogue wave with controllable modulation in its nature. Especially, the nature of rogue wave can be switched between its symmetric and asymmetrically localized excitations. More interestingly, the present rogue wave appear on the top of bell-type (sech) background solitary wave with increased peak amplitude and tail deformation as evidenced from the middle panel of Fig. \ref{rogue-modu}. The single and double kink-type backgrounds (for $k=1$) shall keep the rogue wave unaffected with increased step-like side-bands resembling the simultaneous occurrence of rogue wave with interacting kinks, while an additional periodic function induces asymmetry in peak amplitude and partially degenerates onto the step background, refer to the bottom panel of Fig. \ref{rogue-modu}. To get an important insight of these controllable characteristics and for completeness, we have given a set of graphical demonstrations showcasing the modulation of rational rogue waves due to the periodic, double-periodic, localized, double-localized, and combined type backgrounds in Fig. \ref{rogue-modu} for appropriately chosen choices of parameters as given in the caption.
	\begin{figure}[h]
		\centering
		\includegraphics[width=0.325\linewidth]{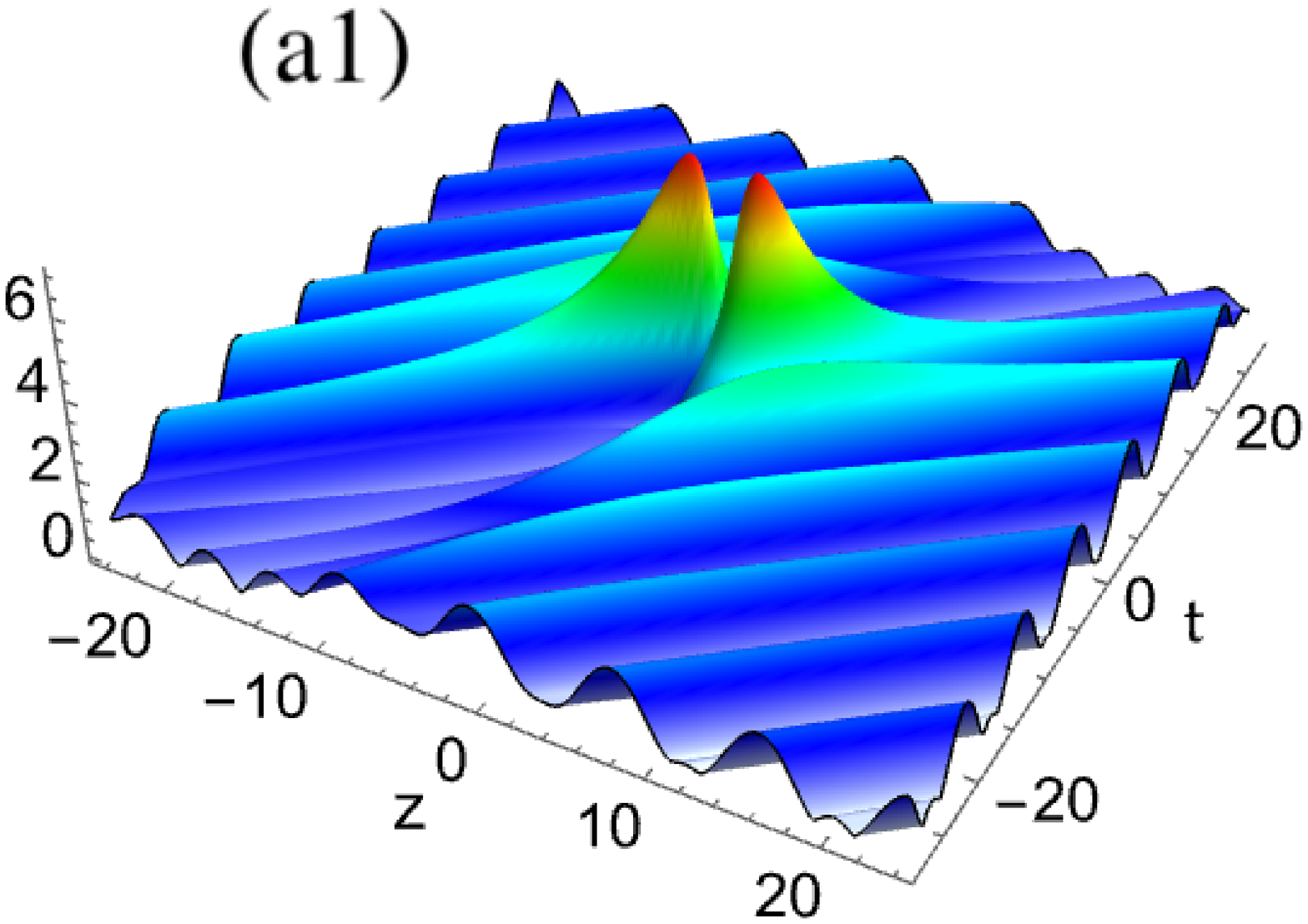}~~\includegraphics[width=0.325\linewidth]{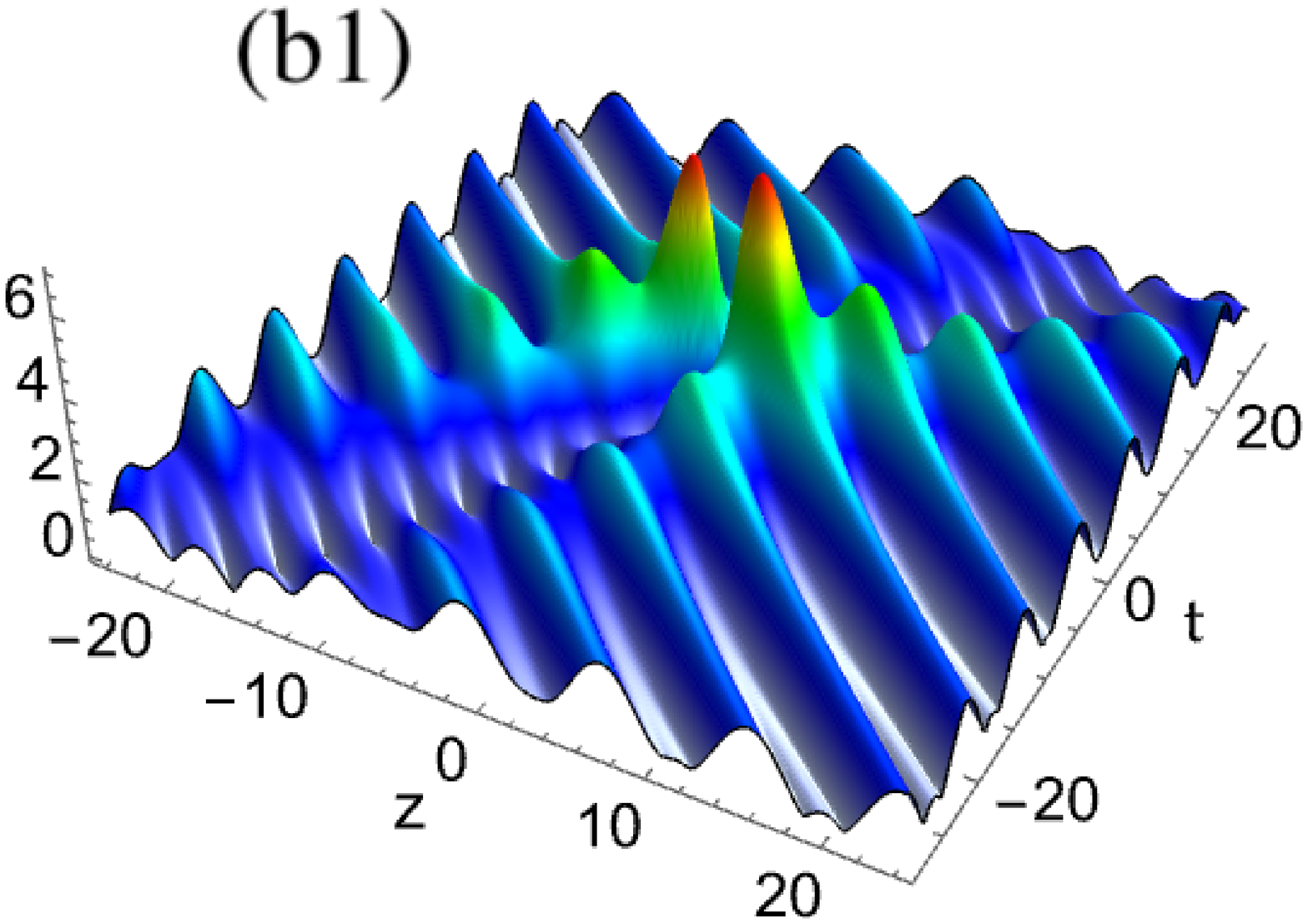}~~\includegraphics[width=0.325\linewidth]{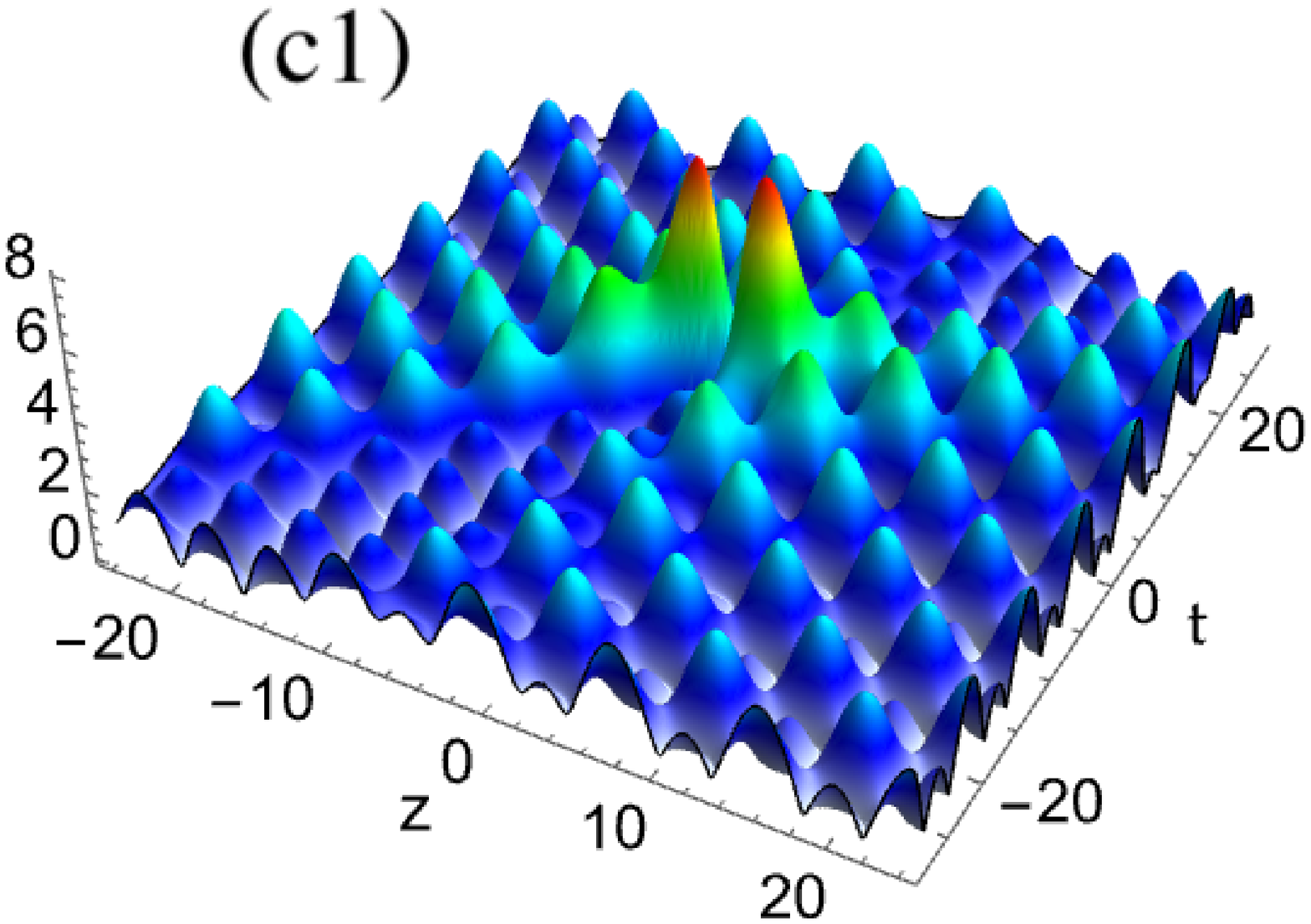}~~\\
		\includegraphics[width=0.325\linewidth]{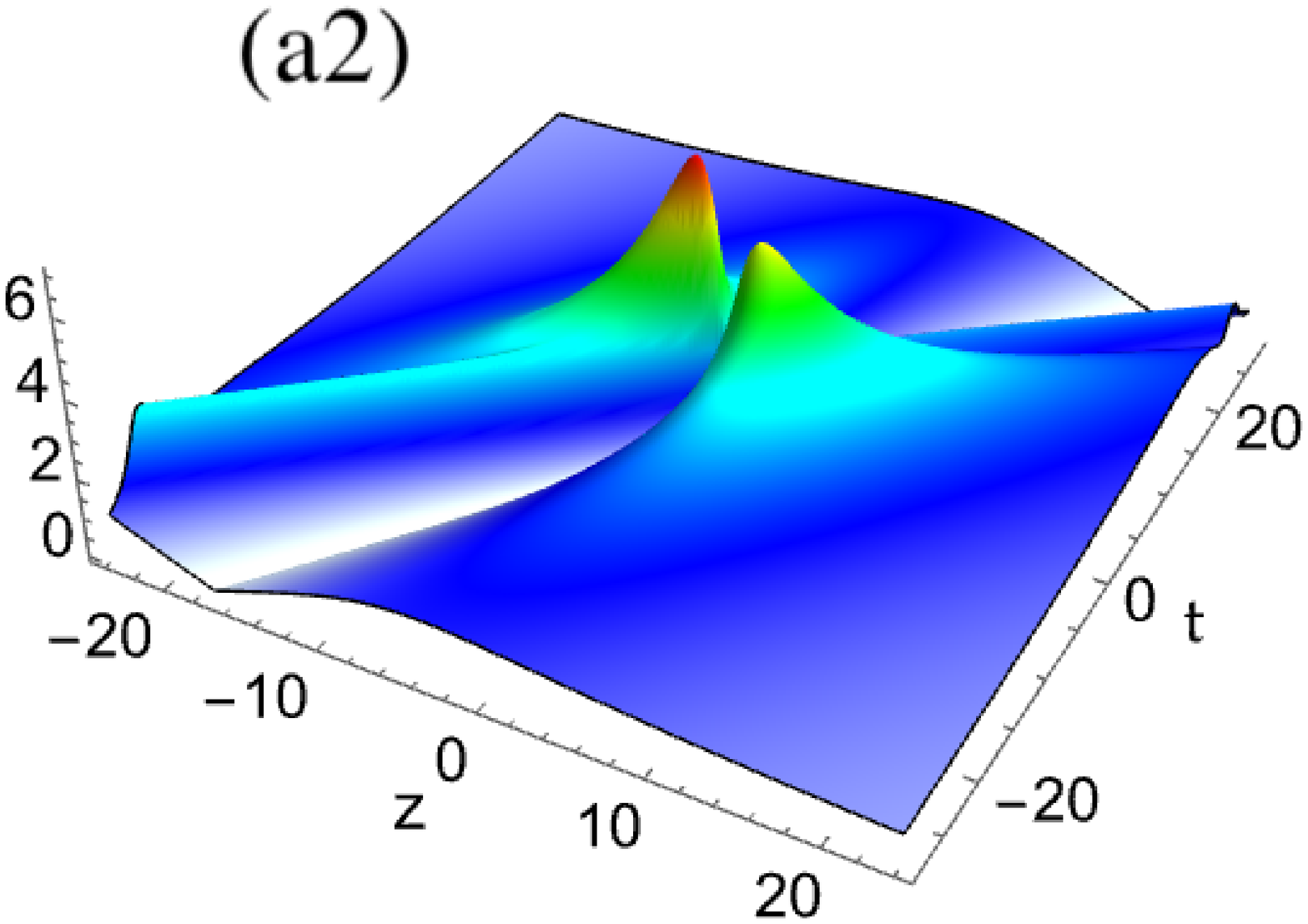}~~\includegraphics[width=0.325\linewidth]{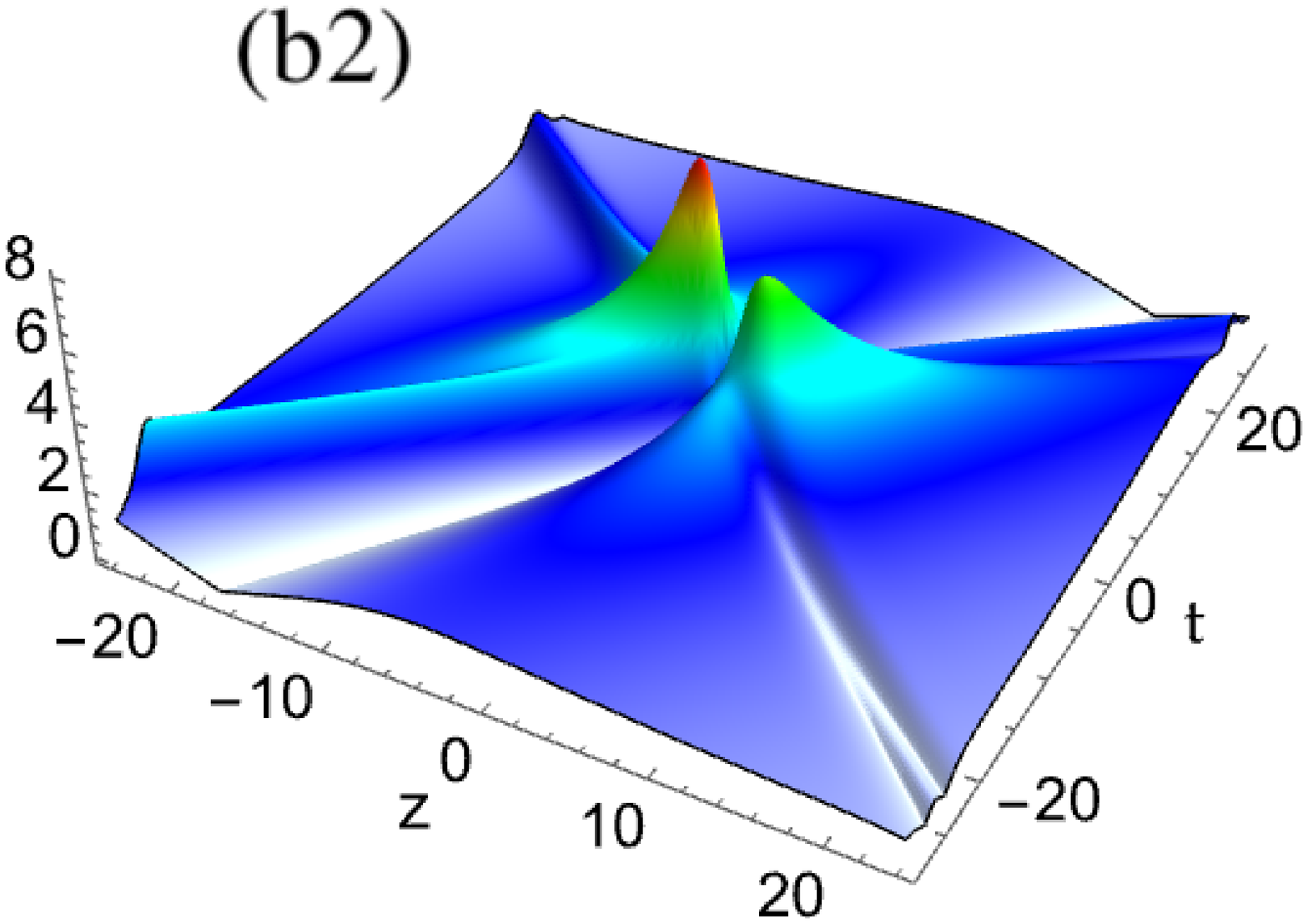}~~\includegraphics[width=0.325\linewidth]{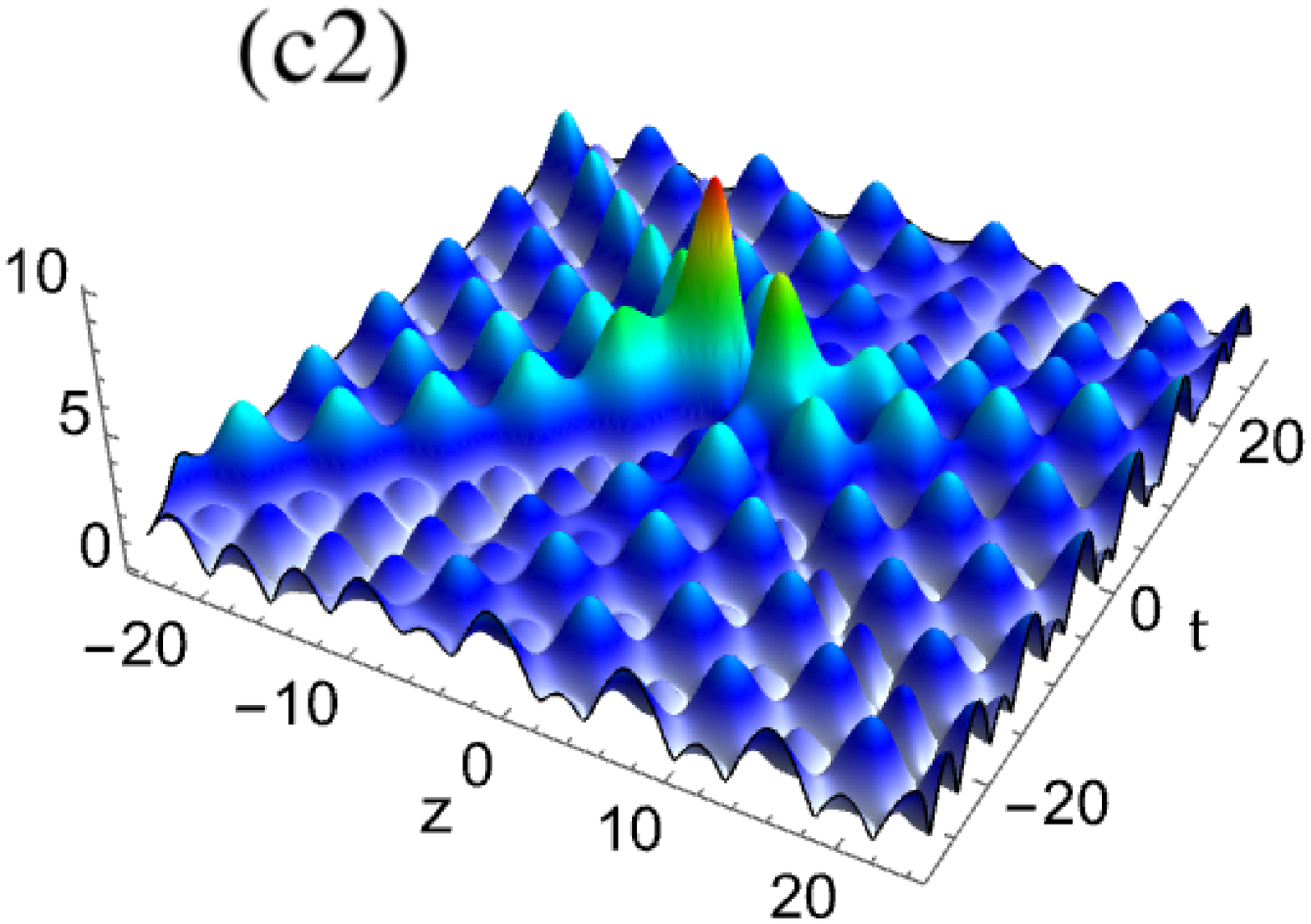}~~\\
		\includegraphics[width=0.325\linewidth]{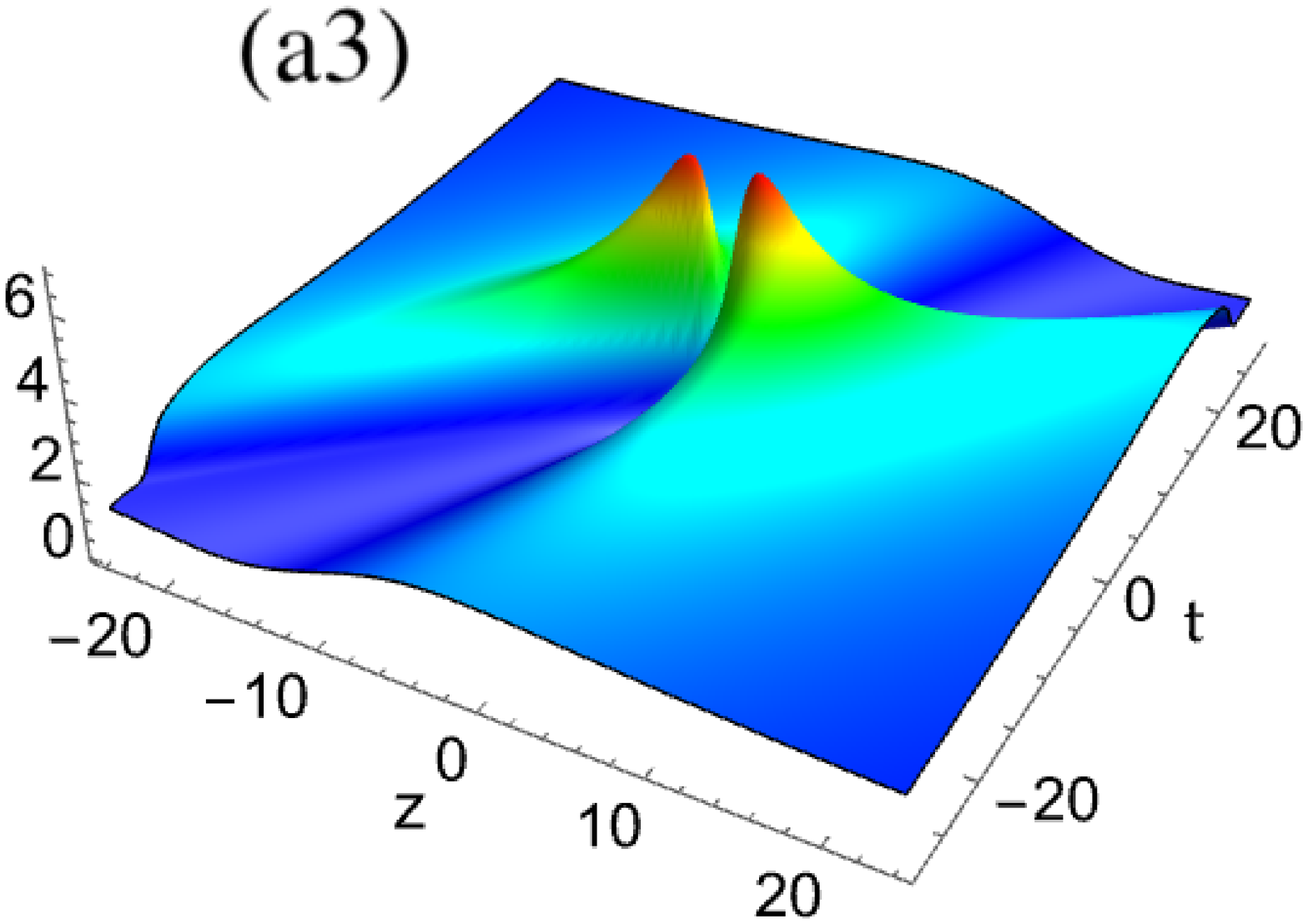}~~\includegraphics[width=0.325\linewidth]{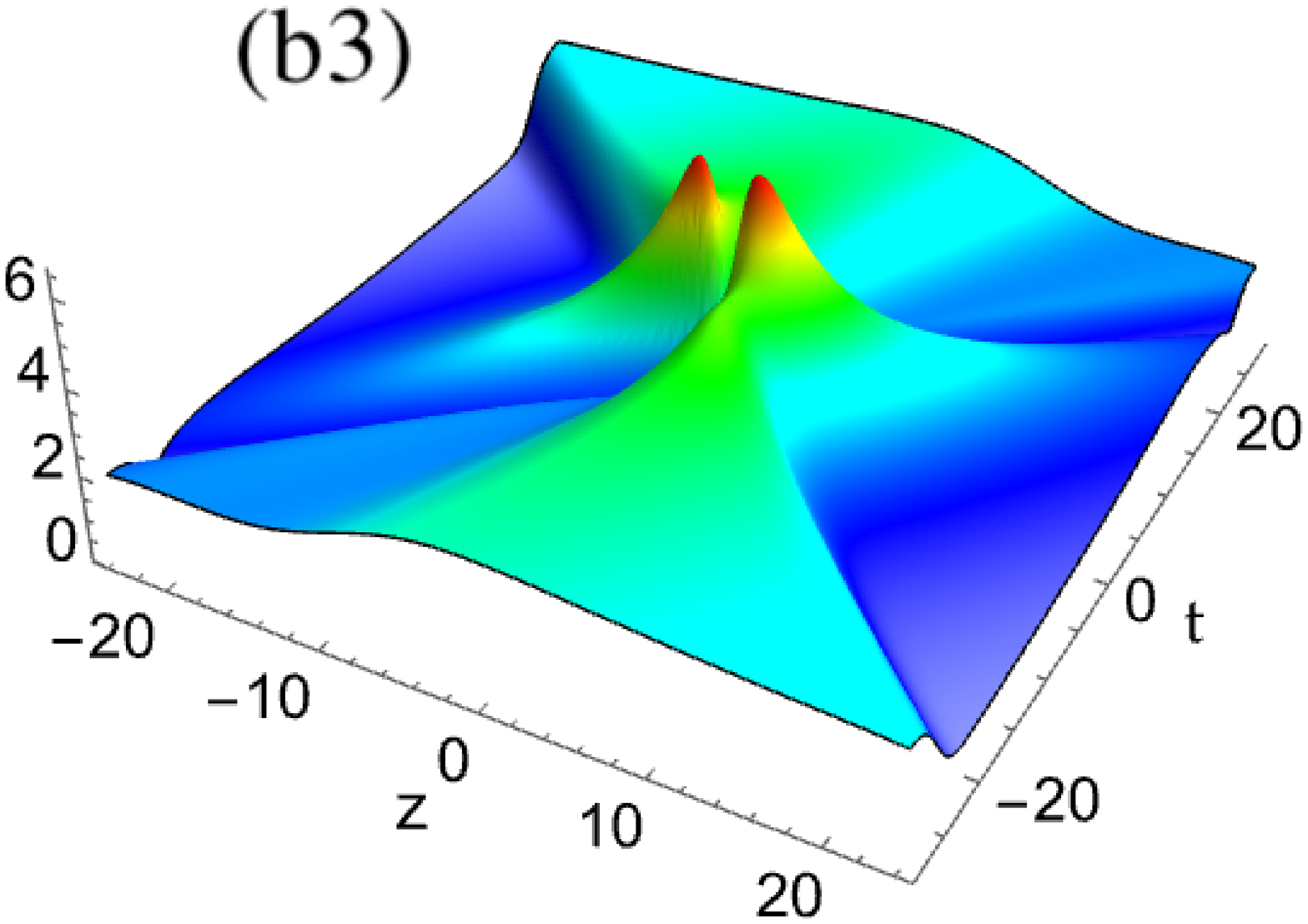}~~\includegraphics[width=0.325\linewidth]{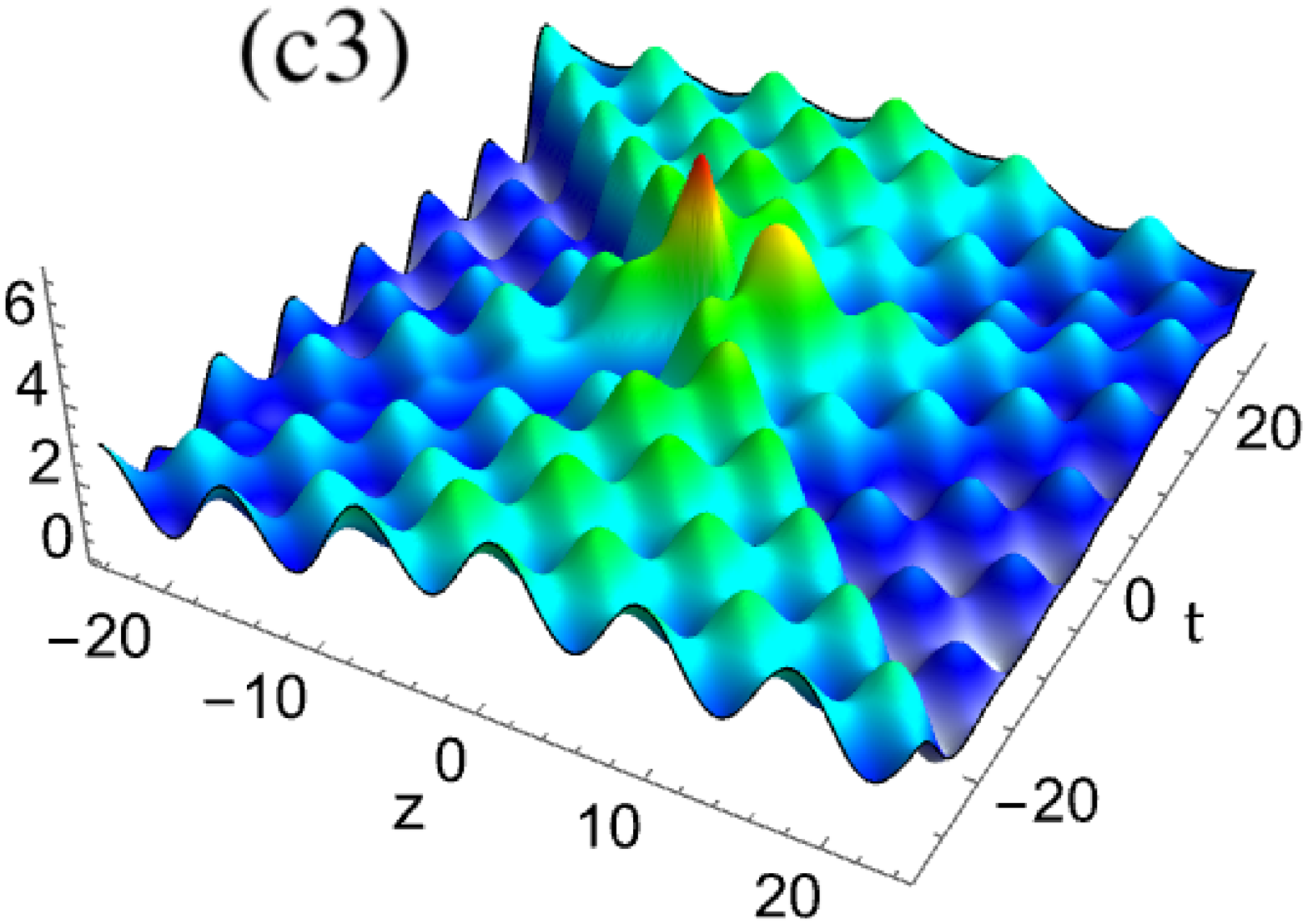}
		\caption{Role of different periodic, localized, and combined wave backgrounds (\ref{jacobi}) on the dynamics of rational rogue wave $|u|$ given by (\ref{eq910}) along $z-t$ plane at $x,y=0.4$ with other parameters as in Fig. \ref{rogue-const}(a,b). The parametric values considered are (a1) $a_1=1.0,a_2=0.75$ and (b1) $b_3=1.0$, $b_4=0.75$ for periodic, (c1) $a_1,b_3=1.0$, $a_2,b_4=0.75$ for double-periodic with $k,m=0$, (a2) $a_3=2.0$, $a_4=0.75$ for single-localized, (b2) $a_3=2.0, b_3=1.0$, $a_4,b_4=0.75$ for double-localized, (c2) $a_1,b_1,b_3=1.0,a_3=2.0$, $a_2,a_4,b_2,b_4=0.75$ for combined with $k=0,m=1$, (a3) $a_1=1.0$, $a_2=0.75$ for a kink, (b3) $a_1,b_1=1.0$, $a_2,b_2=0.75$ for interacting kinks, and (c3) $a_1,b_1=1.0$, $a_3,b_3=0.5$, $a_2,b_2,a_4,b_4=0.75$ for double-periodic interacting kinks with $k=1,m=0$, type backgrounds by taking the other values as zero.} 
		\label{rogue-modu}
	\end{figure}

\section{Results and Discussion} \label{sec4}
{Recently, the impact of nonlinear propagation of periodically varying background (dn- and cn-periodic waves) on the solitons and rogue waves has been observed experimentally in the NLS framework with arbitrarily shaped light waves in optical fibers and a common water wave tank \cite{PRR20,pnas21,FiPexp}. So, it is of potential physical importance to understand such scenario in other nonlinear models too. Also, it is necessary to identify possible mathematical tool that can be instrumental for extending the analysis to a wider class of nonlinear systems. Thus the methodology and results presented in this manuscript serve the required purposes and they will be helpful for characterising the dynamical behaviour of nonlinear waves on different backgrounds.} 
{Particularly, the present route of extracting nonlinear wave solutions has an extra advantage over Darboux transformation and other methods for NLS, KdV, sG, Hirota and their coupled family of equations \cite{epjst,PHD20,RSPA,SiAM,PRE21,  NLS21,PS20, Chaos21,PLA21,CNSNS21,mmas21,PS21,EPJP21,NLD20,CSF19,AML20, RSPA20,AML21,FiP20,PRR21}, in terms of reducing the mathematical/computational complexity as well as a richer variety of solution profiles. To be precise, we have utilized simple exponential and polynomial type test functions as initial seed solutions to obtain the kink soliton (\ref{eq96}) and rogue wave (\ref{eq910}), respectively, which manifested themself to produce various wave phenomena due to the available arbitrary backgrounds (\ref{jacobi}). We have observed the transition from standard kink soliton to periodic solitons (a1 \& b1), double-periodic excitations of kink soliton (c1), elastic (shape-preserving) collisions of bright-kink solitons (a2), double-bright-kink solitons (b2), two-kink solitons (a3), and three-kink solitons (b3) along with collision of periodic solitons (c2 \& c3) as shown in Fig. \ref{kink-peri}. On the other hand, the double-peak rogue wave (\ref{eq910}) revealed an interesting switching possibilities between symmetric and asymmetric amplitude (Fig. \ref{rogue-const}) along with their collisions and coexistence with bright and kink solitons driven by the backgrounds as shown in Figs. \ref{rogue-modu}(a2), (b2), (a3), (b3). Further, the occurrence of rogue wave deformations due to periodic and double-periodic backgrounds also observed Figs. \ref{rogue-modu}(a1), (b1), (c1), (c2), (c3) for suitable choices of arbitrary parameters. For a better understanding on the impact of backgrounds, we have shown the kink soliton and rogue wave without background in Fig. \ref{kink-zero} and Fig. \ref{rogue-const}. The previously reported solutions of (3+1)D KPB model in Refs. \cite{amwt, sun, lkaur, mplb, jpy, lkaur2, wli, ldm,mma21}, especially the soliton solutions \cite{amwt} and rogue waves \cite{sun,mma21} obtained through Hirota bilinear formalism does not have any varying background. However, such controllable background in the present work offered much freedom to manipulate the nonlinear waves accordingly and displayed diverse wave phenomena.}

\textcolor{black}{As mentioned in the introduction, the objectives of the present work:- (i) The search for a simple mathematical tool which explored quite complex wave phenomena. (ii) Dynamics of localized nonlinear waves on sn-, cn-, and dn-periodic backgrounds is studied for the first time to the present model, are successfully achieved. }
\textcolor{black}{ It is worth to highlight that the solutions constructed in the present work through auto-B\"acklund transformation have several additional advantages over the other available methods. Apart from being mathematically simpler and comparatively easier, importantly, the procedure can be applicable for any scalar as well as vector/coupled models in both one- and higher-dimensions. The significance of this approach is further strengthened by its utilizability to construct nonlinear waves on variable backgrounds, which helps the current emerging direction in the study of nonlinear wave dynamics and their manipulation in inhomogeneous media.} 

\textcolor{black}{We are happy to point out certain future directions based on the present study. Apart from the above two types of localized waves and brief discussion on their dynamics, one can extend the present analysis to unravel the features of several other nonlinear wave structures including multiple solitons, general Akhmediev, and Kuznetsov-Ma type breathers, higher-order rogue waves, etc. admitting richer characteristics by appropriately choosing the seed solutions in a straightforward manner. The present work dealt with the effect of background which changes along one spatial dimension along with time coordinate only. So, the inclusion of other two spatial dimensions in the background and investigating the resulting dynamics will be another interesting question. Furthermore, the present analysis can be implemented for various nonlinear wave models, especially higher-dimensional nonlinear equations, to construct localized waves on controllable backgrounds and investigate their underlying dynamics. }

\section{Conclusions}\label{sec5}
{Motivated to study the dynamics of localized nonlinear waves on variable backgrounds in higher-dimensional model, we have considered a (3+1)-dimensional nonlinear equation, referred to Kadomtsev-Petviashvili-Boussinesq model, describing the dynamics of water waves. We have proposed a simple theorem/tool based on the auto-B\"acklund transformation using a truncated Painlev\'e approach and constructed general nonlinear wave solutions along with the significant varying background. By choosing an appropriate initial seed solution we have obtained two types of nonlinear waves of different characteristics, namely kink soliton (\ref{eq96}) and rational rogue wave (\ref{eq910}). The importance of backgrounds (\ref{jacobi}) on such non-trivial structures becomes an exciting option for modulating their dynamics with suitable parameters. We have observed a few features of the physically attractive nature using periodic, localized and combined (periodic+localized) backgrounds chosen through Jacobi elliptic functions which induce significant changes in the dynamics of initial kink soliton and rational (rogue) waves and demonstrated them graphically. This is achieved by implementing the combination of three elliptic functions (sn, cn and dn) as background to the standard nonlinear waves. Particularly, our analysis clearly show the transition of kink soliton to the formation of periodic solitons, double-periodic excitations, shape-preserving collisions with bright or kink soliton(s) (Fig. \ref{kink-peri}) as well as symmetric-asymmetric rogue wave conversion, collisions and coexistence of rogue waves with bright/kink solitons and the deformations of rogue wave (Fig. \ref{rogue-modu}) due to the controllable background as an important factor.} 
The present study can be extended to obtain several other nonlinear wave solutions and to unravel their importance due to backgrounds. The presented theoretical results will be an essential addition to the context of nonlinear waves and stimulate further interest in different wave phenomena in various other nonlinear models both in the theoretical and experimental perspectives.\\

	
\setstretch{1.0}	
	\noindent{\bf Acknowledgements}\\
	One of the authors Sudhir Singh would like to thank the National Institute of Technology Tiruchirappalli and the Ministry of Human Resource Development, Govt. of India, for the financial support through institute fellowship. The research work of K. Sakkaravarthi was supported by the Korean Ministry of Education Science and Technology through Young Scientist Training (YST) Program of the Asia-Pacific Center for Theoretical Physics (APCTP), Pohang-si, Gyeongsangbuk-do.  
	The authors are thankful to the editor and anonymous reviewers for providing valuable comments that immensely improved the manuscript.\\
	
	\noindent{\bf CRediT Authorship Contribution Statement}\\ {\bf Sudhir Singh}: Methodology, Writing - Original Draft Preparation, Writing - Review \& Editing; {\bf K.Sakkaravarthi}: Conceptualization, Validation, Formal Analysis and Investigation, Writing - Original Draft Preparation, Writing - Review \& Editing; {\bf K. Murugesan}: Resources, Writing - Review \& Editing, Funding Acquisition, Supervision.
	
	\noindent{\bf Declaration of Competing Interest}\\ The authors declare that they have no known competing financial interests or personal relationships that could have appeared to influence the work reported in this paper.
	
	
\end{document}